\documentclass{article}

\title{\Hs-algebras and nonunital Frobenius algebras:\\
  first steps in infinite-dimensional\\categorical quantum mechanics}
\author{Samson Abramsky\thanks{Supported by an EPSRC Senior Fellowship and by ONR.} and Chris Heunen\thanks{Supported by the
    Netherlands Organisation for Scientific Research (NWO).}\\
  Oxford University Computing Laboratory}

\usepackage{tikz}
\usetikzlibrary{arrows,positioning,matrix,backgrounds} 
\tikzset{bend angle=45, 
  baseline=(current bounding box.center), 
  inner sep=0ex,
  dot/.style={circle, draw=black, fill=gray, inner sep=.2ex},
  box/.style={inner sep=.3ex, draw, fill=white, rounded corners=.2ex},
  cross/.style={preaction={draw=white, -, line width=3pt}}, 
  every picture/.style=thick, 
  column sep=2ex, 
  row sep=2ex}

\usepackage{xypic} 
\newdir{ >}{{}*!/-7.5pt/@{>}}

\usepackage{amssymb,amsmath} 
\usepackage[thmmarks]{ntheorem}
\theoremnumbering{arabic} 
\theoremstyle{plain}
\theorembodyfont{\itshape} 
\theoremheaderfont{\normalfont\bfseries}
\theoremseparator{} 
\newtheorem{theorem}{Theorem}
\newtheorem{lemma}[theorem]{Lemma}
\newtheorem{proposition}[theorem]{Proposition}

\theorembodyfont{\normalfont}
\newtheorem{definition}[theorem]{Definition}

\theoremstyle{nonumberplain} 
\theoremheaderfont{\scshape}
\newtheorem{proof}{Proof} 
\qedsymbol{\ensuremath{\Box}}
\newcommand{\qedhere}{\tag*{$\Box$}}

\newcommand{\id}[1][]{\ensuremath{\mathrm{id}_{#1}}}
\newcommand{\Id}[1][]{\ensuremath{\mathrm{Id}}}
\newcommand{\after}{\circ}
\newcommand{\cat}[1]{\ensuremath{\mathbf{#1}}}
\newcommand{\Cat}[1]{\ensuremath{\mathbf{#1}}}
\newcommand{\FdHilb}{\Cat{fHilb}} 
\newcommand{\Hilb}{\Cat{Hilb}}
\newcommand{\Rel}{\Cat{Rel}} 
\newcommand{\FRel}{\Cat{fRel}}
\newcommand{\Mat}{\Cat{Mat}}
\newcommand{\LMat}[1][\field{C}]{\Mat_{\ell^2}(#1)}
\newcommand{\LMpos}{\LMat[\Rpos]}
\newcommand{\LBFRel}{\Cat{lbfRel}}
\newcommand{\CQ}{\Cat{canQuant}}
\newcommand{\field}[1]{\ensuremath{\mathbb{#1}}}
\newcommand{\Rpos}{\field{R}^{+}}
\newcommand{\inprod}[2]{\ensuremath{\langle #1\,|\,#2 \rangle}}
\newcommand{\ket}[1]{{|} #1\rangle}
\newcommand{\tensor}{\ensuremath{\otimes}}
\newcommand{\rrel}{\ensuremath{\xymatrix@1@C-2ex{\ar|(.4)@{|}[r]&}}}
\newcommand{\Hs}{H*} 
\newcommand{\ie}{\textit{i.e.}~} 
\newcommand{\eg}{\textit{e.g.}~} 

\begin{document}

\maketitle

\begin{abstract}
  A certain class of Frobenius algebras has been used to characterize
  orthonormal bases and observables on finite-dimensional Hilbert
  spaces. The presence of units in these algebras means that they can
  only be realized finite-dimensionally. We seek a suitable
  generalization, which will allow arbitrary bases, and therefore
  observables with discrete spectra, to be described within categorical
  axiomatizations of quantum mechanics. We develop a definition of
  H*-algebra that can be interpreted in any symmetric monoidal dagger
  category, reduces to the classical notion from functional analysis
  in the category of (possibly infinite-dimensional) Hilbert spaces,
  and hence  provides a categorical way to speak about orthonormal
  bases and quantum observables in arbitrary dimension. Moreover,
  these algebras reduce to the usual notion of Frobenius algebra in
  compact categories. We then investigate the relations between
  nonunital Frobenius algebras and H*-algebras. We give a number of
  equivalent conditions to characterize  when they coincide in the
  category of Hilbert spaces. We also show that they always coincide
  in categories of generalized relations and positive matrices. 
\end{abstract}

\section{Introduction}

The context for this paper comes from the ongoing work on \emph{categorical 
quantum mechanics}~\cite{abramskycoecke:categoricalsemantics,abramskycoecke:categoricalquantummechanics}. This work has
shown how large parts of quantum mechanics can be axiomatized in terms
of monoidal dagger categories and structures definable within them. 
This axiomatization can be used to perform high-level reasoning and
calculations relating to quantum information, using diagrammatic
methods~\cite{selinger:graphicallanguages}; and also as a basis for
exploring foundational issues in quantum mechanics and quantum
computation. In particular, a form of Frobenius algebras has been
used to give an algebraic axiomatization of \emph{orthonormal bases}
and \emph{observables}~\cite{coeckepavlovic:classicalobjects,
coeckepavlovicvicary:bases}. 

The structures used so far (\eg compact closure, Frobenius algebras)
have only finite-dimensional realizations in Hilbert spaces.  
This raises some interesting questions and challenges:
\begin{itemize}
  \item Find a good general notion of Frobenius structure which works
    in the infinite-dimensional case in $\Hilb$.
  \item Use this to characterize general bases and therefore general
    observables with discrete spectra.
  \item Similarly extend the analysis for other categories.
  \item Clarify the mathematics, and relate it to the wider literature. 
\end{itemize}
As we shall see, an intriguing problem remains open, but
much of this program of work has been accomplished.  

The further contents of the paper are as
follows. Section~\ref{sec:background} recalls some background on
monoidal dagger categories and Frobenius algebras, and poses the
problem. Section~\ref{sec:hstar} introduces the key notion of 
\Hs-algebra, in the general setting of symmetric monoidal dagger
categories. In Section~\ref{sec:hstarinhilb}, we prove our results
relating to $\Hilb$, the category of Hilbert spaces (of unrestricted
dimension). We show how \Hs-algebras provide exactly the right
algebraic notion to characterize orthonormal bases in arbitrary
dimension. We give several equivalent characterizations of when
\Hs-algebras and nonunital Frobenius algebras coincide in the category
of Hilbert spaces. Section~\ref{sec:hstarinrel} studies \Hs-algebras
in categories of generalized relations and positive matrices. We show
that in these settings, where no phenomena of `destructive
interference' arise, \Hs-algebras and nonunital Frobenius algebras
always coincide. Finally, Section~\ref{sec:outlook} provides an
outlook for future work.  

\section{Background}
\label{sec:background}

The basic setting is that of \emph{dagger symmetric monoidal
categories}. We briefly recall the definitions,
referring to~\cite{abramskycoecke:categoricalquantummechanics} for
further details and motivation.  

A \emph{dagger category} is a category $\cat{D}$ equipped with an
identity-on-objects, contravariant, strictly involutive
functor. Concretely, for each arrow $f \colon A \to B$, there is an arrow
$f^\dag \colon B \to A$, and this assignment satisfies
\[ 
  \id^\dag = \id, \qquad 
  (g \circ f)^\dag = f^\dag \circ g^\dag, \qquad 
  f^{\dag\dag} = f \, . 
\] 
An arrow $f \colon A \to B$ is \emph{dagger monic} when $f^\dag
\circ f = \id[A]$, and a \emph{dagger iso(morphism)} if both $f$ and
$f^\dag$ are dagger monics. 

A \emph{symmetric monoidal dagger category} is a dagger category with
a symmetric monoidal structure $(\cat{D}, \tensor, I, \lambda,
\rho, \alpha, \sigma)$ such that 
\[ 
  (f \otimes g)^\dag = f^\dag \otimes g^\dag
\]
and moreover the natural isomorphisms $\lambda$, $\rho$, $\alpha$,
$\sigma$ are componentwise dagger isomorphisms. 

\subsubsection*{Examples}
\label{ex:cats}
  %Let us give some examples of such categories, that we will be
  %concerned with.
  \begin{itemize}
  \item The category $\Hilb$ of Hilbert spaces and continuous linear
    maps, and its (full) subcategory $\FdHilb$ of finite-dimensional
    Hilbert spaces. Here the dagger is the adjoint, and the tensor
    product has its standard interpretation for Hilbert spaces.
    Dagger isomorphisms are \emph{unitaries}, and dagger monics are
    \emph{isometries}.
  \item The category $\Rel$ of sets and relations. Here the dagger is
    relational converse, while the monoidal structure is given by the
    cartesian product. This generalizes to relations valued in a
    commutative quantale~\cite{rosenthal:quantales}, and to the
    category of relations of any regular
    category~\cite{carbonikasangianstreet:relations}. This has a full sub-category $\FRel$, of finite sets and relations.
  \item The category $\LBFRel$, of \emph{locally bifinite
    relations}. This is the subcategory of $\Rel$ comprising those
    relations which are image-finite, meaning that each element in the
    domain is related to only finitely many elements in the codomain,
    and whose converses are also image-finite. 
    This forms a monoidal dagger subcategory of $\Rel$. It serves as
    a kind of qualitative approximation of the passage from finite- to
    infinite-dimensional Hilbert spaces. For example, a set carries a
    compact structure in $\LBFRel$ if and only if it is finite. 
  \item A common generalization of $\FdHilb$ and $\FRel$ is obtained by forming the category
    $\Mat(S)$, where $S$ is a commutative semiring with a specified
    involution~\cite{heunen:hilbcatsembedding}. Objects of $\Mat(S)$ 
    are finite sets, and morphisms are maps \mbox{$X \times Y \to S$}, which
    we think of as `$X$ times $Y$ matrices'. Composition is by matrix
    multiplication, while the dagger is conjugate transpose, where
    conjugation of a matrix means elementwise application of the
    involution on $S$. The tensor product of $X$ and $Y$ is given by
    $X \times Y$, with the action on matrices given by componentwise
    multiplication, corresponding to the `Kronecker product' of
    matrices. If we take $S = \field{C}$, this yields a category
    equivalent to $\FdHilb$, while taking $S$ to be the Boolean
    semiring $\{0, 1\}$, with trivial involution, gives $\FRel$.
  \item An infinitary generalization of $\Mat(\field{C})$ is given by
    $\LMat$. This category has arbitrary sets as objects, and its
    morphisms $X \to Y$ are matrices \mbox{$M \colon X \times Y \to \field{C}$} such that
    for each $x \in X$, the family $\{ M(x, y) \}_{y \in Y}$ is
    $\ell^2$\-summable; and for each $y \in Y$, the family $\{ M(x, y)
    \}_{x \in X}$ is $\ell^2$\-summable. The category $\Hilb$ is
    equivalent to a (nonfull) subcategory of
    $\LMat$~\cite[Theorem~3.1.7]{blanketal:hilbert}.  
  \end{itemize}
%\end{example}

\subsubsection*{Graphical Calculus}
We briefly recall the graphical calculus for symmetric monoidal dagger
categories \cite{selinger:graphicallanguages}. 
This can be seen as a two-dimensional version of
\emph{Dirac notation}, which allows equational reasoning
to be performed graphically in a sound and complete fashion.
A morphism $f \colon X \to Y$ is represented pictorially as
$
 \begin{tikzpicture}[scale=0.5,font=\scriptsize]%
  \node (X) at (0,-.75) {$X$};
\node (Y) at (0,.75) {$Y$};
\draw (X) to (Y);
\node [box] at (0,0) {$\:f\:$};
 \end{tikzpicture}%
$, the identity 
on $X$ simply becomes
$
 \begin{tikzpicture}[thick,scale=0.5,font=\scriptsize]%
  \node (X) at (0,-.5) {$X$};
\node (Y) at (0,.5) {$X$};
\draw (X) to (Y);
 \end{tikzpicture}%
$, and composition
and tensor products  appear as follows. 
\[
  
 \begin{tikzpicture}[scale=0.5]%
  \node (X) at (0,-1.5) {$X$};
\node (Z) at (0,1.5) {$Z$};
\draw (X) to (Z);
\node [box] at (0,0) {$\;g \after f\;$};
 \end{tikzpicture}%

  \;=\;
  
 \begin{tikzpicture}[scale=0.5]%
  \node (X) at (0,-1.5) {$X$};
\node (Y) at (0,0) {};
\node (Z) at (0,1.5) {$Z$};
\draw (X) to (Z);
\node [box] at (0,0.5) {$\;g\vphantom{f}\;$};
\node [box] at (0,-0.5) {$\;f\;$};
 \end{tikzpicture}%

  \qquad\qquad
  
 \begin{tikzpicture}[scale=0.5]%
  \node (WX) at (0,-1.5) {$W \tensor X$};
\node (YZ) at (0,1.5) {$Y \tensor Z$};
\draw (WX) to (YZ);
\node [box] at (0,0) {$\;f \tensor g\;$};
 \end{tikzpicture}%

  \;=\;
  
 \begin{tikzpicture}[scale=0.5]%
  \node (W) at (-.5,-1.5) {$W$};
\node (Y) at (-.5,1.5) {$Y$};
\draw (W) to (Y);
\node [box] at (-.5,0) {$\;f\;$};
\node (X) at (.75,-1.5) {$X$};
\node (Z) at (.75,1.5) {$Z$};
\draw (X) to (Z);
\node [box] at (.75,0) {$\;g\vphantom{f}\;$};
 \end{tikzpicture}%

\]
The symmetry isomorphism $\sigma$ is drawn as
$\vcenter{\hbox{
 \begin{tikzpicture}[scale=0.25]%
  \node (0) at (-1,1) {};
\node (1) at (1,1) {};
\node (2) at (-1,-1) {};
\node (3) at (1,-1) {};
\draw[in=90, out=270] (0) to (3);
\draw[cross, in=90, out=270] (1) to (2);
 \end{tikzpicture}%
}}$.
The dagger is represented graphically by a horizontal
reflection. 

\subsection{Dagger Frobenius algebras}

Frobenius algebras are a classic notion in
mathematics~\cite{nakayama:frobenius}. A particular form of such
algebras was introduced in the general setting of monoidal dagger
categories by Coecke and Pavlovi{\'c}
in~\cite{coeckepavlovic:classicalobjects}. In their version, a
\emph{dagger Frobenius structure} on an object $A$ in a dagger
monoidal category is a commutative comonoid 
$(\xymatrix@1{I & A \ar|-{\varepsilon}[l] \ar|-{\delta}[r] & A \tensor
A})$ satisfying certain additional equations:  
\begin{align}
      (\id[A] \otimes \delta) \after \delta 
  & = (\delta \otimes \id[A]) \after \delta, \tag{A} \\
      (\id[A] \tensor \varepsilon) \after \delta 
  & = \id[A], \tag{U} \\
      \sigma \after \delta 
  & = \delta, \tag{C} \\
      \delta^\dag \after \delta 
  & = \id[A], \tag{M} \\
      \delta\after \delta^\dag 
  & = (\delta^\dag \tensor \id[A]) \after (\id[A] \tensor \delta). \tag{F}
\end{align}
These equations become more perspicuous when represented
diagrammatically, as below. Here, we draw the comultiplication
$\delta$ as $\vcenter{\hbox{
 \begin{tikzpicture}[scale=0.2]%
  \node (3) at (-1,2) {};
\node [style=dot] (1) at (0,1) {};
\node (0) at (0,0) {};
\node (2) at (1,2) {};
\draw [bend right=45] (1) to (2);
\draw [bend right=45] (3) to (1);
\draw (1) to (0);
 \end{tikzpicture}%
}}$,
and the counit $\varepsilon$ as
$\vcenter{\hbox{
 \begin{tikzpicture}[scale=0.2]%
  \node (0) at (0,0) {};
\node [style=dot] (1) at (0,1) {};
\draw (0) to (1);
 \end{tikzpicture}%
}}$. 
\[
  
 \begin{tikzpicture}[]%
  \matrix[matrix of math nodes]
{    |(A)| &            & |(B)|      & |(C)| \\
           & |(D)[dot]| &            & |(E)| \\
           &            & |(F)[dot]|         \\
           &            & |(G)|              \\
};
\draw (F) to (G);
\draw (A) to [bend right] (D);
\draw (D) to [bend right] (B);
\draw (D) to [bend right] (F);
\draw (F) to [bend right] (E) -- (C);
 \end{tikzpicture}%
 \; \stackrel{\text{(A)}}{=} \; 
 \begin{tikzpicture}[]%
  \matrix[matrix of math nodes]
{    |(A)| & |(B)| &       & |(C)| \\
     |(D)| &       & |(E)[dot]| \\
	   & |(F)[dot]| \\
	   & |(G)| \\
};
\draw (B) to [bend right] (E);
\draw (E) to [bend right] (C);
\draw (E) to [bend left] (F);
\draw (F) to [bend left] (D) -- (A);
\draw (F) to (G);
 \end{tikzpicture}%

  \qquad\qquad\qquad
  
 \begin{tikzpicture}[]%
  \matrix[matrix of math nodes]
{ |(C)| &            & |(D)[dot]| \\
	& |(E)[dot]|              \\
	& |(F)|                   \\
};
\draw (C) to [bend right] (E);
\draw (E) to [bend right] (D);
\draw (E) to (F);
 \end{tikzpicture}%
 \; \stackrel{\text{(U)}}{=} \; 
 \begin{tikzpicture}[]%
  \matrix[matrix of math nodes]{ |(A)| \\ \\ \\ |(B)| \\};
\draw (A) to (B);
 \end{tikzpicture}%

  \qquad\qquad\qquad
  
 \begin{tikzpicture}[]%
  \matrix[matrix of math nodes]
{ |(A)| &            & |(B)| \\ 
  |(C)| &            & |(D)| \\ 
	& |(E)[dot]|         \\
	& |(F)|              \\
};
\draw (E) to (F);
\draw (E) to [bend right] (D) .. controls (B) and (C) .. (A);
\draw[cross] (E) to [bend left] (C) .. controls (A) and (D) .. (B);
 \end{tikzpicture}%
 \; \stackrel{\text{(C)}}{=} \; 
 \begin{tikzpicture}[]%
  \matrix[matrix of math nodes]
{ |(C)| &            & |(D)| \\
	& |(E)[dot]|         \\
	& |(F)|              \\
};
\draw (C) to [bend right] (E);
\draw (E) to [bend right] (D);
\draw (E) to (F);
 \end{tikzpicture}%

\]
\[ 
  
 \begin{tikzpicture}[]%
  \matrix[matrix of math nodes]
{       & |(A)|              \\
	& |(B)[dot]|         \\
  |(C)| &            & |(D)| \\
	& |(E)[dot]|         \\
	& |(F)|              \\
};
\draw (A) to (B);
\draw (B) to [bend right] (C) to [bend right] (E);
\draw (B) to [bend left] (D) to [bend left] (E);
\draw (E) to (F);
 \end{tikzpicture}%
 \; \stackrel{\text{(M)}}{=} \; 
 \begin{tikzpicture}[]%
  \matrix[matrix of math nodes]{ |(A)| \\ \\ \\ \\ |(B)| \\};
\draw (A) to (B);
 \end{tikzpicture}%

  \qquad\qquad\qquad
  
 \begin{tikzpicture}[]%
  \matrix[matrix of math nodes]
{          & |(A)|      &       &            & |(B)| \\
	 & |(C)[dot]| &       &                    \\
   |(D)| &            & |(E)| &            & |(F)| \\
	 &            &       & |(G)[dot]|         \\
   |(H)| &            &       & |(I)|              \\
};
\draw (C) to [bend right] (D) -- (H);
\draw (C) to (A);
\draw (C) to [bend left] (E) to [bend right] (G);
\draw (G) to (I);
\draw (G) to [bend right] (F) -- (B);
 \end{tikzpicture}%
 \; \stackrel{\text{(F)}}{=} \; 
 \begin{tikzpicture}[]%
  \matrix[matrix of math nodes]
{    |(A)| &            & |(B)| \\
	 & |(C)[dot]|         \\ 
	 & |(D)[dot]|         \\         
   |(E)| &            & |(F)| \\
};
\draw (A) to [bend right] (C);
\draw (C) to [bend right] (B);
\draw (C) to (D);
\draw (E) to [bend left] (D);
\draw (D) to [bend left] (F);
 \end{tikzpicture}%

\]
A `right-handed version' of the Frobenius law (F) follows from (C); in
the noncommutative case we should add this symmetric version (F') to
axiom (F). 

\subsection{Dagger Frobenius algebras in quantum mechanics}

Frobenius algebras provide a high-level algebraic way of talking about
\emph{orthonormal bases}, and hence can be seen as modeling quantum
mechanical \emph{observables}.  

To put this in context, we recall the \emph{no-cloning
theorem}~\cite{wootterszurek:nocloning}, which says that 
there is no quantum evolution (\ie unitary operator) $f 
\colon H \to H \tensor H$ such that, for any $\ket{\phi} \in H$,  
\[ 
  f\ket{\phi} = \ket{\phi} \tensor\ket{\phi}. 
\]
A general form of no-cloning holds for structural reasons in
categorical quantum mechanics~\cite{abramsky:nocloning}. In
particular, there is no \emph{natural}, \ie uniform or 
basis-independent, family of diagonal morphisms in a compact closed
category, unless the category collapses, so that endomorphisms are
scalar multiples of the identity.  

However, if we drop naturality, we \emph{can} define such maps in
$\Cat{Hilb}$ in a basis-dependent fashion. Moreover, it turns out that such
maps can be used to \emph{uniquely determine} bases. 
Firstly, consider \emph{copying maps}, which can be defined in
arbitrary dimension: for a given basis $\{\ket{i}\}_{i \in I}$ of $H$,
define $\delta \colon H \to H \tensor H$ by (continuous linear
extension of) $\ket{i} \mapsto \ket{ii}$.

For example, consider the map $\delta_{\text{std}} \colon \field{C}^2 \to
\field{C}^2 \tensor \field{C}^2$ defined by
\[ \ket{0} \mapsto \ket{00}, \qquad \ket{1} \mapsto \ket{11} . \]
By construction, this copies the elements of the
computational basis --- and  \emph{only} these, as in general
\[
  \delta_{\text{std}}(\alpha\ket{0}+\beta\ket{1}) =
  \alpha\ket{00}+\beta\ket{11} \neq (\alpha\ket{0}+\beta\ket{1})
  \tensor (\alpha\ket{0}+\beta\ket{1}).
\]
Next, consider \emph{deleting maps} $\varepsilon \colon H \to
\field{C}$ by linearly extending $\ket{e_i} \mapsto 1$. In
contrast to copying, these can be defined in \emph{finite dimension only}. 
It is straightforward to verify that these maps define a
dagger Frobenius structure on $H$. Moreover, the following result
provides a striking converse. 

\begin{theorem}
  \label{cpvth}
  \cite{coeckepavlovicvicary:bases}
  Orthonormal bases of a finite-dimensional Hilbert space $H$ are in one-to-one
  correspondence with dagger Frobenius structures on $H$.
  \qed
\end{theorem}

This result in fact follows easily from previous results in the literature on Frobenius algebras \cite{abrams:thesis}; we will give a short proof from
the established literature in Section~\ref{subsec:furtherconditions}.    

Another result provides a counterpart---at first sight displaying very
different looking behaviour---in the category $\Rel$.  

\begin{theorem}
  \cite{pavlovic:frobeniusinrel}
  Dagger Frobenius structures in the category $\Rel$ correspond
  to disjoint unions of abelian groups. 
  \qed
\end{theorem}

We shall provide a different proof of this result in Section~\ref{frobrelsec}, which makes no use of units, and hence generalizes to a wide range of other situations, such as locally bifinite and quantale-valued relations, and positive $\ell_2$-matrices.

\subsection{The problem}

The notion of Frobenius structure as defined above, which requires a
unit, limits us to the \emph{finite-dimensional case} in $\Hilb$, as the following lemma shows. 

\begin{lemma}
  A Frobenius algebra in $\Cat{Hilb}$ is unital if and only if it is
  finite-dimensional. 
\end{lemma}
\begin{proof}
  Sufficiency is shown in~\cite[3.6.9]{kock:frobenius}. Necessity
  follows from~\cite[Corollary to Theorem 4]{kaplansky:dualrings}.
\end{proof}

In fact, a Frobenius structure on an object $A$ induces a
\emph{compact} (or \emph{rigid}) structure on $A$, with $A$ as its own
dual (see~\cite{abramskycoecke:categoricalquantummechanics}). Indeed,
put \mbox{$\eta = \delta \after \varepsilon^\dag \colon 
I \to A \tensor A$}. In the category $\FdHilb$, for example, $\eta
\colon \field{C} \to \field{C}^2 \tensor \field{C}^2$ is an
\emph{entangled state preparation}:
\[ 
    \eta_{\text{std}} 
  = \delta_{\text{std}}\after \varepsilon^\dag_{\text{std}} 
  = (1 \mapsto \delta_{\text{std}}(\ket{0}+\ket{1}))
  = (1 \mapsto \ket{00} + \ket{11}). 
\]
In general it is easy to see that $\eta$ indeed provides a dagger
compact structure on $A$, with $A^*=A$:
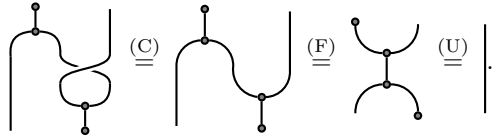
\[
  
 \begin{tikzpicture}[scale=0.33]%
  \node [style=dot] (0) at (-1, 3) {};
\node (1) at (2, 3) {};
\node [style=dot] (2) at (-1, 2) {};
\node (3) at (-2, 1) {};
\node (4) at (0, 1) {};
\node (5) at (2, 1) {};
\node (6) at (0, -0) {};
\node (7) at (2, -0) {};
\node [style=dot] (8) at (1, -1) {};
\node (9) at (-2, -2) {};
\node [style=dot] (10) at (1, -2) {};
\draw[looseness=0.75, out=-90, in=90] (4.north) to (7.south);
\draw[looseness=1.25, bend left=45] (2) to (4);
\draw (5) to (1);
\draw[looseness=1.25, bend left=45] (3.center) to (2);
\draw[looseness=1.25, bend right=45] (8) to (7);
\draw[cross, looseness=0.75, out=90, in=270] (6.south) to (5.north);
\draw (8) to (10);
\draw[looseness=1.25, bend right=45] (6) to (8);
\draw (9.center) to (3.center);
\draw (2) to (0);
 \end{tikzpicture}%

  \; \stackrel{\text{(C)}}{=} \;
  
 \begin{tikzpicture}[scale=0.5]%
  \matrix[matrix of math nodes]
{        & |(A)[dot]| &       &            & |(B)| \\
	 & |(C)[dot]| &       &                    \\
   |(D)| &            & |(E)| &            & |(F)| \\
	 &            &       & |(G)[dot]|         \\
   |(H)| &            &       & |(I)[dot]|         \\
};
\draw (C) to [bend right] (D) -- (H);
\draw (C) to (A);
\draw (C) to [bend left] (E) to [bend right] (G);
\draw (G) to (I);
\draw (G) to [bend right] (F) -- (B);
 \end{tikzpicture}%

  \; \stackrel{\text{(F)}}{=} \;
  
 \begin{tikzpicture}[scale=0.5]%
  \matrix[matrix of math nodes]
{  |(A)[dot]| &            & |(B)|      \\
    	      & |(C)[dot]|              \\ 
	      & |(D)[dot]|              \\         
   |(E)|      &            & |(F)[dot]| \\
};
\draw (A) to [bend right] (C);
\draw (C) to [bend right] (B);
\draw (C) to (D);
\draw (E) to [bend left] (D);
\draw (D) to [bend left] (F);
 \end{tikzpicture}%

  \; \stackrel{\text{(U)}}{=} \;
  
 \begin{tikzpicture}[scale=0.5]%
  \matrix[matrix of math nodes]{  |(A)| \\ \\ \\ \\ |(B)| \\ };
\draw (A) to (B);
 \end{tikzpicture}%
.
\]
As is well-known, a compact structure  exists only
for finite-dimensional spaces in $\Cat{Hilb}$. Thus to obtain a notion capable of being
extended beyond the finite-dimensional case, we need to drop the
assumption of a unit.  

%. But,
%of course, a large part of quantum mechanics is concerned with
%infinite-dimensional Hilbert spaces. The goal of the present work is to
%characterize categorical structures similar to Frobenius algebras to
%play the role of observables in the infinite-dimensional case. Units
%present an immediate obstacle

\section{\Hs-algebras}
\label{sec:hstar}

We begin our investigation of suitable axioms for a notion of algebra
which can characterize orthonormal bases in arbitrary dimension by
recalling the axioms for Frobenius structures.
\begin{align}
      (\id[A] \otimes \delta) \after \delta 
  & = (\delta \otimes \id[A]) \after \delta \tag{A} \\
      (\id[A] \tensor \varepsilon) \after \delta 
  & = \id[A] \tag{U} \\
      \sigma \after \delta 
  & = \delta \tag{C} \\
      \delta^\dag \after \delta 
  & = \id[A] \tag{M} \\
      \delta\after \delta^\dag 
  & = (\delta^\dag \tensor \id[A]) \after (\id[A] \tensor \delta) & \tag{F}
\end{align}
We note in passing that there is some redundancy in the definition of
Frobenius structure. 

\begin{lemma}
  In any dagger monoidal category, (M), (F) and (F') imply (A). 
\end{lemma}
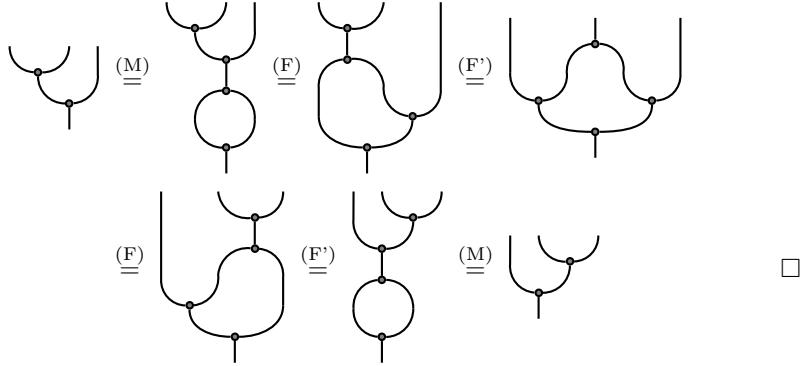
\begin{proof}
 % First, four applications of (F) yield
  % \begin{align*}
  %       (\delta \tensor \id) \after \delta \after \delta^\dag
  %   & = (\delta \tensor \id) \after (\delta^\dag \after \id) \after
  %       (\id \tensor \delta) \\
  %   & = (\id \tensor \delta^\dag \tensor \id) \after (\delta \tensor
  %       \delta) \\
  %   & = (\id \tensor \delta) \after (\id \tensor \delta^\dag) \after
  %       (\delta \tensor \id) \\
  %   & = (\id \tensor \delta) \after \delta \after \delta^\dag.
  % \end{align*}
  % So by (M), we have $(\delta \tensor \id) \after \delta
  % = (\id \tensor \delta) \after \delta$.
  \begin{align*}
    
 \begin{tikzpicture}[scale=0.5]%
  \matrix[matrix of math nodes]
{    |(A)| &            & |(B)|      & |(C)| \\
           & |(D)[dot]| &            & |(E)| \\
           &            & |(F)[dot]|         \\
           &            & |(G)|              \\
};
\draw (F) to (G);
\draw (A) to [bend right] (D);
\draw (D) to [bend right] (B);
\draw (D) to [bend right] (F);
\draw (F) to [bend right] (E) -- (C);
 \end{tikzpicture}%

    & \; \stackrel{\text{(M)}}{=} \;
    
 \begin{tikzpicture}[scale=0.5]%
  \matrix[matrix of math nodes]
{    |(A)| &            & |(B)|      & |(C)| \\
           & |(D)[dot]| &            & |(E)| \\
           &            & |(F)[dot]|         \\
           &            & |(G)[dot]|              \\
           & |(H)| &            & |(I)| \\
           &            & |(J)[dot]|              \\
           &            & |(K)|              \\
};
\draw (F) to (G);
\draw (A) to [bend right] (D);
\draw (D) to [bend right] (B);
\draw (D) to [bend right] (F);
\draw (F) to [bend right] (E) -- (C);
\draw (G) to [bend right] (H.center);
\draw (H.center) to [bend right] (J);
\draw (G) to [bend left] (I.center);
\draw (I.center) to [bend left] (J);
\draw (J) to (K);
 \end{tikzpicture}%

    \; \stackrel{\text{(F)}}{=} \;
    
 \begin{tikzpicture}[scale=0.5]%
  \matrix[matrix of math nodes, column sep=1ex]
{    |(A)| &&                 && |(B)|      &&        &&  |(C)| \\
             && |(D)[dot]| \\
             && |(E)[dot]| \\
     |(F)| &&                 && |(G)|     &&          && |(H)| \\
     |(I)| &&                  &&             && |(J)[dot]| \\
             &&& |(K)[dot]| \\
	     &&& |(L)| \\
};
\draw (A) to [bend right] (D);
\draw (D) to [bend right] (B);
\draw (D) to (E);
\draw (I.center)--(F.center) to [bend left] (E);
\draw (E) to [bend left] (G) to [bend right] (J);
\draw (J) to [bend right] (H) -- (C);
\draw (I.center) to [out=270, in=180] (K);
\draw (K) to [out=0, in=270] (J);
\draw (K) to (L);
 \end{tikzpicture}%

    \; \stackrel{\text{(F')}}{=} \;
    
 \begin{tikzpicture}[scale=0.5]%
  \matrix[matrix of math nodes]
{ |(A)|  & & & |(B)| & & & |(C)| \\
    & & & |(L)[dot]| \\
  |(D)| & & |(E)| & & |(F)| & & |(G)| \\
  & |(H)[dot]| & & & & |(I)[dot]| \\
  & & & |(J)[dot]| \\
  & & & |(K)| \\
};
\draw (A) to (D.center);
\draw (D.center) to [bend right] (H);
\draw (H) to [bend right] (E) to [bend left] (L);
\draw (B) to (L);
\draw (L) to [bend left] (F) to [bend right] (I);
\draw (I) to [bend right] (G)--(C);
\draw (H) to [out=270, in=180] (J);
\draw (J) to [out=0, in=270] (I);
\draw (J) to (K);
 \end{tikzpicture}%
 \\
    & \; \stackrel{\text{(F)}}{=} \;
    
 \begin{tikzpicture}[scale=0.5]%
  \matrix[matrix of math nodes, column sep=1ex]
{    |(C)| &&                 && |(B)|      &&        &&  |(A)| \\
            &&                  &&             && |(D)[dot]| \\
            &&                  &&             && |(E)[dot]| \\
     |(H)| &&                 && |(G)|     &&          && |(F)| \\
              && |(J)[dot]| &&           &&             && |(I)| \\
              &&                  &&& |(K)[dot]| \\
	     && &&& |(L)| \\
};
\draw (A) to [bend left] (D);
\draw (D) to [bend left] (B);
\draw (D) to (E);
\draw (I.center)--(F.center) to [bend right] (E);
\draw (E) to [bend right] (G) to [bend left] (J);
\draw (J) to [bend left] (H) -- (C);
\draw (I.center) to [out=270, in=0] (K);
\draw (K) to [out=180, in=270] (J);
\draw (K) to (L);
 \end{tikzpicture}%

    \; \stackrel{\text{(F')}}{=} \;
    
 \begin{tikzpicture}[scale=0.5]%
  \matrix[matrix of math nodes]
{    |(A)| & |(B)| &       & |(C)| \\
     |(D)| &       & |(E)[dot]| \\
	   & |(F)[dot]| \\
                       & |(G)[dot]|              \\
            |(H)| &            & |(I)| \\
                       & |(J)[dot]|              \\
                       & |(K)|              \\
};
\draw (B) to [bend right] (E);
\draw (E) to [bend right] (C);
\draw (E) to [bend left] (F);
\draw (F) to [bend left] (D) -- (A);
\draw (F) to (G);
\draw (G) to [bend right] (H.center);
\draw (H.center) to [bend right] (J);
\draw (G) to [bend left] (I.center);
\draw (I.center) to [bend left] (J);
\draw (J) to (K);
 \end{tikzpicture}%

    \; \stackrel{\text{(M)}}{=} \;
    
 \begin{tikzpicture}[scale=0.5]%
  \matrix[matrix of math nodes]
{    |(A)| & |(B)| &       & |(C)| \\
     |(D)| &       & |(E)[dot]| \\
	   & |(F)[dot]| \\
	   & |(G)| \\
};
\draw (B) to [bend right] (E);
\draw (E) to [bend right] (C);
\draw (E) to [bend left] (F);
\draw (F) to [bend left] (D) -- (A);
\draw (F) to (G);
 \end{tikzpicture}%

  \qedhere
  \end{align*}
\end{proof}
The axioms (U), (C), and (F) are independent: 
\begin{itemize}
 \item As we have seen, for an orthonormal basis $\{\ket{n} \mid n \in
   \field{N}\}$ of a separable (infinite-dimensional) Hilbert space,
   the map $\delta(\ket{n}) = \ket{nn}$ satisfies everything except
   for (U).
 \item Group algebras of finite noncommutative
   groups~\cite[Example~4]{ambrose:hstar} satisfy everything except 
   for (C).
\item Any nontrivial commutative (unital) Hopf algebra satisfies everything except for
  (F) 
  by~\cite[Proposition~2.4.10]{kock:frobenius}. 
\end{itemize}
It is worth noting that under additional assumptions, such as unitality
and enrichment in abelian groups, (A) and (M) are known to imply
(F)~\cite[Section~6]{longoroberts:dimension}. 

We shall now \emph{redefine} a Frobenius algebra\footnote{In the
literature the unital version is more
specifically termed a special commutative dagger Frobenius
algebra (sometimes also called a separable algebra, or a Q-system). As
we will only be concerned with these kinds of Frobenius algebras, we
prefer to keep terminology simple and dispense with the adjectives.}
in a dagger monoidal category to be an object $A$ equipped with a
comultiplication $\delta : A \to A \tensor A$ satisfying (A), (C), (M)
and (F). A Frobenius algebra which additionally has an arrow
$\varepsilon \colon A \to I$ satisfying (U) will explicitly be called
\emph{unital}.  

\subsection{Regular representation as pointwise abstraction}

As we have seen, unital Frobenius algebras allow us to define
compact, and hence closed, structure. How much of this can we keep  in
key examples such as $\Cat{Hilb}$? 

The category $\Cat{Hilb}$ has well-behaved duals, since $H \cong
H^{**}$, and indeed there is a conjugate-linear isomorphism $H \cong
H^{*}$. However, it is \emph{not} the case that the tensor unit
$\field{C}$ is exponentiable in $\Cat{Hilb}$, since if it was, we
would have a bounded linear evaluation map 
\[ 
  H \tensor H^* \to \field{C},
\]
and hence its adjoint $\field{C} \to H \tensor H^*$, and a compact
structure. 

We shall now present an axiom which captures what seems to
be the best we can do in general in the way of a `transfer of
variables'. It is, indeed, a general form, meaningful in any 
monoidal dagger category, of a salient structure in functional analysis.  

Suppose we have a comultiplication $\delta \colon A
\to A \tensor A$, and hence a multiplication $\mu = \delta^\dag
\colon A \tensor A \to A$. We can \emph{curry} the multiplication
(this process is also called
$\lambda$-abstraction~\cite{barendregt:lambdacalculus}) for
\emph{points}---this is just the regular representation!\footnote{As
we are in a commutative context, there is no need to distinguish
between left and right regular representations.} Thus we have
a function $R \colon \cat{D}(I,A) \to \cat{D}(A,A)$ defined by
\[
  R(a) = \mu \after(\id \tensor a) = 
 \begin{tikzpicture}[scale=0.5]%
  \matrix[matrix of math nodes]
{                & |(A)|                            \\
		 & |(B)[dot]|                       \\
  |(C)| {\phantom{a}} &            & |(D)[box]| {a} \\
};
\draw (A) to (B);
\draw (B) to [out=0, in=90] (D);
\draw (B) to [bend right] (C.north) -- (C.south);
 \end{tikzpicture}%
.
\]
If $\mu$ is associative, this is a semigroup homomorphism. 

\subsection{Axiom (H)}

An endomorphism homset $\cat{D}(A,A)$ in a dagger category $\cat{D}$ is not
just a monoid, but a \emph{monoid with involution}, because of the dagger.
We say that $(A, \mu)$ \emph{satisfies axiom (H)} if there is an operation $a
\mapsto a^*$ on $\cat{D}(I, A)$ such that $R$ becomes a homomorphism
of involutive semigroups, \ie 
\[
  R(a^*) = R(a)^\dag
\]
for every $a \colon I \to A$. This unfolds to
\begin{align}
    \mu \after (a^* \tensor \id) 
  = (a^\dag \tensor \id) \after \mu^\dag; \tag{H}
\end{align}
or diagrammatically:
\[
  
 \begin{tikzpicture}[scale=0.5]%
  \matrix[matrix of math nodes]
{                  & |(A)|                             \\
		 & |(B)[dot]|                        \\
|(C)| {\phantom{a^*}} &            & |(D)[box]| {a^*} \\
};
\draw (A) to (B);
\draw (B) to [out=180, in=90] (C.north) -- (C.south);
\draw (B) to [out=0, in=90] (D);
 \end{tikzpicture}%

  \; \stackrel{\text{(H)}}{=} \;
  
 \begin{tikzpicture}[scale=0.5]%
  \matrix[matrix of math nodes]
{ |(A)| {\phantom{a^\dag}} &            & |(B)[box]| {a^\dag} \\
		   & |(C)[dot]|                           \\
		   & |(D)|                                \\
};
\draw (C) to [out=0, in=270] (B);
\draw (C) to [out=180, in=270] (A.south) -- (A.north);
\draw (C) to (D);
 \end{tikzpicture}%
.
\]
Thus $a \mapsto a^*$ is indeed a `transfer of variables'. 

% \begin{definition}
% \label{def:hstar}
%   An \emph{\Hs-algebra}\footnote{To be precise, we should call the
%   structure in Definition~\ref{def:hstar} a proper commutative
%   H*-algebra. But as with Frobenius algebras, we prefer to keep
%   terminology simple.} in a dagger braided monoidal category $\cat{D}$
%   consists of a morphism $\delta \colon A \to A \tensor A$ together
%   with an involution $* \colon \cat{D}(I,A)\op \to \cat{D}(I,A)$
%   satisfying (A), (C), (M) and (H). When there exists a morphism
%   $\varepsilon$ satisfying (U), we explicitly call the algebra
%   \emph{unital}. 
% \end{definition}

\subsection{Relationships between axioms (F) and (H)}

The rest of this section compares axioms (F) and (H) at the
abstract level of monoidal dagger categories. 
%We have already seen that both
%enable some transfer of variables, which indicates some
%interdependency. 

The following observation by Coecke, Pavlovi{\'c} and 
Vicary is the central idea in their proof of Theorem~\ref{cpvth}.

\begin{lemma}
\label{lem:cpv}
  In any dagger monoidal category, (F) and (U) imply (H).
\end{lemma}
\begin{proof}
  Define $a^* = (a^\dag \tensor \id) \after \delta \after \varepsilon^\dag$. 
  \[
    
 \begin{tikzpicture}[]%
  \matrix[matrix of math nodes]{ |(A)| \\ |(B)[box]| {a^*} \\ };
\draw (A) to (B);
 \end{tikzpicture}%
  \;=\;  
 \begin{tikzpicture}[]%
  \matrix[matrix of math nodes]
{ |(A)| {\phantom{a^\dag}} &            & |(B)[box]| {a^\dag} \\
		  & |(C)[dot]|                           \\
		  & |(D)[dot]|                           \\ 
};
\draw (D) to (C);
\draw (C) to [out=0, in=270] (B);
\draw (C) to [out=180, in=270] (A.south) -- (A.north);
 \end{tikzpicture}%

  \]
  This indeed satisfies (H).
  \[
    
 \begin{tikzpicture}[]%
  \matrix[matrix of math nodes]
{                  & |(A)|                             \\
		 & |(B)[dot]|                        \\
|(C)| {\phantom{a^*}} &            & |(D)[box]| {a^*} \\
};
\draw (A) to (B);
\draw (B) to [out=180, in=90] (C.north) -- (C.south);
\draw (B) to [out=0, in=90] (D);
 \end{tikzpicture}%

    \; = \;
    
 \begin{tikzpicture}[]%
  \matrix[matrix of math nodes]
{          & |(A)|      &       &            & |(B)[box]|{a^\dag} \\
	 & |(C)[dot]| &       &                    \\
   |(D)| &            & |(E)| &            & |(F)| \\
	 &            &       & |(G)[dot]|         \\
   |(H)| &            &       & |(I)[dot]|         \\
};
\draw (C) to [bend right] (D) -- (H);
\draw (C) to (A);
\draw (C) to [bend left] (E) to [bend right] (G);
\draw (G) to (I);
\draw (G) to [out=0, in=270] (F) -- (B);
 \end{tikzpicture}%

    \; \stackrel{\text{(F)}}{=} \;
    
 \begin{tikzpicture}[]%
  \matrix[matrix of math nodes]
{    |(A)|{\phantom{a^\dag}} &            & |(B)[box]|{a^\dag} \\
		     & |(C)[dot]|                          \\ 
		     & |(D)[dot]|                          \\         
   |(E)|        &            & |(F)[dot]|                  \\
};
\draw (C) to [out=180, in=270] (A.south) -- (A.north);
\draw (C) to [out=0, in=270] (B);
\draw (C) to (D);
\draw (E) to [out=90, in=180] (D);
\draw (D) to [out=0, in=90] (F);
 \end{tikzpicture}%

    \; \stackrel{\text{(U)}}{=} \;
    
 \begin{tikzpicture}[]%
  \matrix[matrix of math nodes]
{ |(A)| {\phantom{a^\dag}} &            & |(B)[box]| {a^\dag} \\
		   & |(C)[dot]|                           \\
		   & |(D)|                                \\
};
\draw (C) to [out=0, in=270] (B);
\draw (C) to [out=180, in=270] (A.south) -- (A.north);
\draw (C) to (D);
 \end{tikzpicture}%

  \]
\end{proof}

\noindent Recall that a category is \emph{monoidally well-pointed} if the
following holds:
\begin{align*}
 f = g \colon A \tensor A' \to B \tensor B' 
  \;\;\; \Longleftrightarrow \;\;\; 
  \forall x \colon I \to A, y \colon I \to A' .\,\,f \after (x \tensor y) = g \after (x \tensor y).
\end{align*}
All the categories listed in our Examples are
monoidally well-pointed in this sense. 

\begin{lemma}
\label{lem:HAWPimplyF}
  In a monoidally well-pointed dagger monoidal category, (H) and (A) imply (F).
\end{lemma}
\begin{proof}
  For any $a \colon I \to A$ we have the following.
  \[ 
    
 \begin{tikzpicture}[]%
  \matrix[matrix of math nodes]
{    |(A)|{\phantom{a^\dag}} &            & |(B)[box]|{a^\dag} \\
		     & |(C)[dot]|                          \\ 
		     & |(D)[dot]|                          \\         
   |(E)|             &            & |(F)|                  \\
};
\draw (A.north) -- (A.south) to [out=270, in=180] (C);
\draw (C) to [out=0, in=270] (B);
\draw (C) to (D);
\draw (E) to [out=90, in=180] (D);
\draw (D) to [out=0, in=90] (F);
 \end{tikzpicture}%

    \; \stackrel{\text{(H)}}{=} \;
    
 \begin{tikzpicture}[]%
  \matrix[matrix of math nodes]
{   &              & |(A)|                           \\
		& & |(B)[dot]|                      \\
& |(C)|{\phantom{a^*}} &            & |(D)[box]|{a^*}\\
		            & |(E)[dot]|         \\
		|(F)|      &            & |(G)| \\
};
\draw (A) to (B);
\draw (B) to [out=180, in=90] (C.north) -- (E);
\draw (B) to [out=0, in=90] (D);
\draw (E) to [bend right] (F);
\draw (E) to [bend left] (G);
 \end{tikzpicture}%

    \; \stackrel{\text{(A)}}{=} \;
    
 \begin{tikzpicture}[]%
  \matrix[matrix of math nodes]
{                            & |(A)|               \\
		            & |(B)[dot]|          \\
		 |(E)| &            & |(D)[dot]|  \\
 |(F)|{\phantom{a^*}}            & |(G)|{\phantom{a^*}} & & |(H)[box]|{a^*} \\
};
\draw (A) to (B);
\draw (B) to [out=180, in=90] (E.south) -- (F.south);
\draw (B) to [out=0, in=90] (D);
\draw (D) to [out=180, in=90] (G.north) -- (G.south);
\draw (D) to [out=0, in=90] (H);
 \end{tikzpicture}%

    \; \stackrel{\text{(H)}}{=} \;
    
 \begin{tikzpicture}[]%
  \matrix[matrix of math nodes]
{          & |(A)|      &       &            & |(B)[box]|{a^\dag} \\
	 & |(C)[dot]| &       &                    \\
   |(D)| &            & |(E)| &            & |(F)| \\
	 &            &       & |(G)[dot]|         \\
   |(H)| &            &       & |(I)|              \\
};
\draw (C) to [bend right] (D) -- (H);
\draw (C) to (A);
\draw (C) to [bend left] (E) to [bend right] (G);
\draw (G) to (I);
\draw (G) to [out=0, in=270] (F) -- (B);
 \end{tikzpicture}%

  \]
  Then (F) follows from monoidal well-pointedness.
\end{proof}

Lemma~\ref{lem:cpv} is strengthened by the following proposition,
which proves that compactness  implies unitality.

\begin{proposition}
\label{prop:cptimpliesunital}
  Any Frobenius algebra in a dagger compact category is unital.
\end{proposition}
\begin{proof}
 \cite[Remark (1) on page 503]{carboni:matrices}
 Suppose that $\delta \colon A \to A \tensor A$ is a nonunital
 Frobenius algebra in a compact category. Define $\varepsilon \colon A
 \to I$ as follows.
 \[
  \varepsilon \;=\; 
 \begin{tikzpicture}[]%
  \node (0) at (0,0) {};
\node [style=dot] (1) at (0,1) {};
\draw (0) to (1);
 \end{tikzpicture}%
 \;=\; 
 \begin{tikzpicture}[]%
  \matrix[matrix of math nodes]
{ & |(D)[dot]| & & & |(C)| \\
  |(E)| & & |(F)| & & |(G)| \\
  & & |(H)| & & |(I)| \\
  |(J)| & & & |(K)| \\ 
};
\draw (D) to [bend right] (E) -- (J);
\draw (D) to [out=90, in=90] (C) to (G.south);
\draw (D) to [bend left] (F)
	  .. controls (H) and (G) .. (I.south) to [out=90, in=90] (H);
\draw[cross] (I) to [out=270, in=270] (H)
	  .. controls (F) and (I) .. (G);
 \end{tikzpicture}%

 \]
 Then the following holds, where we draw the unit and counit of
 compactness by caps and cups (without dots).
 
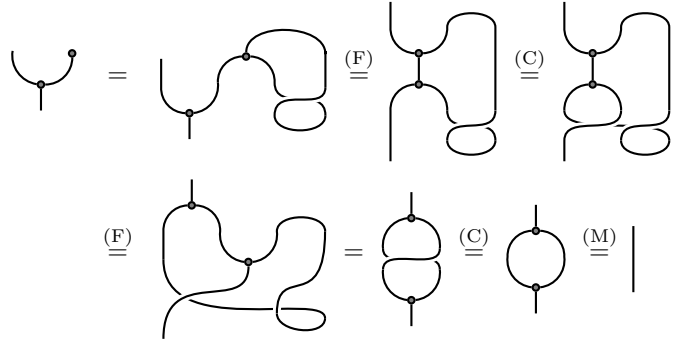
\begin{eqnarray*}
     
 \begin{tikzpicture}[]%
  \matrix[matrix of math nodes]{ |(A)| & & |(B)[dot]| \\ & |(C)[dot]| \\ & |(D)| \\};
\draw (D) to (C);
\draw (C) to [bend left] (A.north);
\draw (C) to [bend right] (B);
 \end{tikzpicture}%

   & = 
   & 
 \begin{tikzpicture}[]%
  \matrix[matrix of math nodes]
{ |(a)| && & |(D)[dot]| & & & |(C)| \\
  |(b)| && |(E)| & & |(F)| & & |(G)| \\
  & |(c)[dot]| & & & |(H)| & & |(I)| \\
  & |(J)| & & & & |(K)| \\ 
};
\draw (D) to [bend right] (E);
\draw (D) to [out=90, in=90] (C) to (G.south);
\draw (D) to [bend left] (F)
       .. controls (H) and (G) .. (I.south) to [out=90, in=90] (H);
\draw[cross] (I) to [out=270, in=270] (H)
      .. controls (F) and (I) .. (G);
\draw (c) to (J);
\draw (c) to [bend left] (b) -- (a);
\draw (c) to [bend right] (E.north);
 \end{tikzpicture}%

     \; \stackrel{\text{(F)}}{=} \;
     
 \begin{tikzpicture}[]%
  \matrix[matrix of math nodes]
{ |(A)| & & |(D)| & & |(E)| \\
  |(F)| & & |(G)| & & |(e)| \\
  & |(H)[dot]| \\
  & |(I)[dot]| \\
  |(J)| & & |(K)| && |(L)| \\
  & & |(k)| && |(l)| \\
  |(M)| \\
};
\draw (H) to (I);
\draw (I) to [bend right] (J) -- (M);
\draw (H) to [bend left] (F) -- (A);
\draw (I) to [bend left] (K);
\draw (H) to [bend right] (G) to [out=90, in=90] (e.south) -- (L);
\draw (k) to [out=270, in=270] (l);
\draw (K.north) .. controls (k) and (L) .. (l.south);
\draw[cross] (k.south) .. controls (K) and (l) .. (L.north);

 \end{tikzpicture}%

     \; \stackrel{\text{(C)}}{=} \;
     
 \begin{tikzpicture}[]%
  \matrix[matrix of math nodes]
{ |(A)| & & |(D)| & & |(E)| \\
  |(F)| & & |(G)| & & |(e)| \\
  & |(H)[dot]| \\
  & |(I)[dot]| \\
  |(J)| & & |(K)| && |(L)| \\
  |(j)| & & |(k)| && |(l)| \\
  |(M)| & |(m)| \\
};
\draw (H) to (I);
\draw (I) to [bend right] (J);
\draw (H) to [bend left] (F) -- (A);
\draw (I) to [bend left] (K);
\draw (H) to [bend right] (G) to [out=90, in=90] (e.south) -- (L);
\draw (k) to [out=270, in=270] (l);
\draw (J.north) .. controls (j) and (L) .. (l.south);
\draw[cross] (k.south) .. controls (K) and (l) .. (L.north);
\draw[cross] (K.north) .. controls (k) and (J) .. (j.south) -- (M);
 \end{tikzpicture}%
 \\
   & \stackrel{\text{(F)}}{=} 
   & 
 \begin{tikzpicture}[]%
  \matrix[matrix of math nodes]
{ & |(A)| \\
  & |(B)[dot]| & & & |(C)| && |(D)| \\
  |(E)| & & |(F)| & & |(G)| && |(g)| \\
  |(h)| & & & |(H)[dot]| & |(i)| && |(I)| \\
  &&&& |(J)| && |(K)| \\
  &&& |(l)| & |(L)| && |(M)| \\
  |(a)| \\
};
\draw (A) to (B);
\draw (B) to [bend right] (E) -- (h);
\draw (B) to [bend left] (F) to [bend right] (H);
\draw (H) to [bend right] (G.north) to [out=90, in=90](g)--(I);
\draw (L.north) [out=270, in=270] to (M);
\draw (h.north) .. controls (a) and (K) .. (M.south);
\draw[cross] (g.north) .. controls (M) and (i) .. (L);
\draw[cross] (H) .. controls (l) and (h) .. (a);
 \end{tikzpicture}%

     \; = \;
     
 \begin{tikzpicture}[]%
  \matrix[matrix of math nodes]
{ & |(a)| \\
  & |(A)[dot]| \\
  |(B)| && |(C)| \\
  |(D)| && |(E)| \\
  & |(F)[dot]| \\
  & |(f)| \\
};
\draw (a) to (A);
\draw (A) to [bend right] (B.south);
\draw (A) to [bend left] (C.south);
\draw (F) to [bend left] (D.north);
\draw (F) to [bend right] (E.north);
\draw (F) to (f);
\draw (B) .. controls (D) and (C) .. (E);
\draw[cross] (D) .. controls (B) and (E) .. (C);
 \end{tikzpicture}%

     \; \stackrel{\text{(C)}}{=} \;
     
 \begin{tikzpicture}[]%
  \matrix[matrix of math nodes]
{ & |(a)| \\
  & |(A)[dot]| \\
  |(B)| && |(C)| \\
  & |(F)[dot]| \\
  & |(f)| \\
};
\draw (a) to (A);
\draw (A) to [bend right] (B) to [bend right] (F);
\draw (A) to [bend left] (C) to [bend left] (F);
\draw (F) to (f);
 \end{tikzpicture}%

     \; \stackrel{\text{(M)}}{=} \;
     
 \begin{tikzpicture}[]%
  \matrix[matrix of math nodes]{ |(A)| \\ \\ \\ |(B)| \\ };
\draw[thick] (A) to (B);
 \end{tikzpicture}%

 \end{eqnarray*}
 That is, (U) holds.
\end{proof}

Thus in the unital, monoidally well-pointed case, (F) and (H) are essentially
equivalent. Our interest is, of course, in the nonunital case. To
explain the provenance of the (H) axiom, and its implications for
obtaining a correspondence with orthonormal bases in $\Cat{Hilb}$ in
arbitrary dimension, we shall now study the situation in the
concrete setting of Hilbert spaces.

\section{\Hs-algebras in $\Cat{Hilb}$}
\label{sec:hstarinhilb}

We begin by revisiting Theorem~\ref{cpvth}. How should the
correspondence between Frobenius algebras and orthonormal bases be
expressed mathematically? In fact, the content of this result is
really a \emph{structure theorem} of a classic genre in algebra
\cite{albert:structureofalgebras}. The following theorem, the
\emph{Wedderburn structure theorem}, is the prime example; it was subsequently
generalized by Artin, and there have been many subsequent developments. 

\begin{theorem}[Wedderburn, 1908]
  Every finite-dimensional semisimple algebra is isomorphic to a
  product of full matrix algebras. In the commutative case over the
  complex numbers, this has the form: the algebra is isomorphic to a
  product of one-dimensional complex algebras. 
  \qed
\end{theorem}

To see the connection between the Wedderburn structure theorem and 
Theorem~\ref{cpvth}, consider the coalgebra $A$ determined by an
orthonormal basis $\{ \ket{i} \}$ on a Hilbert space:  
\[ 
  \delta \colon \ket{i} \mapsto \ket{ii}.
\]
This is isomorphic as a coalgebra to a direct sum of one-dimensional
coalgebras 
\[ 
  \delta_{\field{C}} \colon \field{C} \to \field{C} \tensor \field{C}, 
  \qquad 1 \mapsto 1 \tensor 1.
\]
To say that a Frobenius algebra corresponds to an orthonormal basis
is exactly to say that it is isomorphic as a coalgebra to a Hilbert
space direct sum of one-dimensional coalgebras: 
\[ 
  A \; \cong \; \bigoplus_I \, (\field{C}, \delta_{\field{C}}),
\]
where the cardinality of $I$ is the dimension of $H$.
Applying dagger, this is equivalent to $A$ being isomorphic as an \emph{algebra} to
the direct sum of one-dimensional \emph{algebras}  
\[ 
  A \; \cong \; \bigoplus_I \, (\field{C}, \mu_{\field{C}}), 
  \qquad \mu_{\field{C}} \colon 1 \tensor 1 \mapsto 1.
\]
In this case, we say that the Frobenius algebra \emph{admits the
structure theorem}, making the view of bases as (co)algebras precise. 

\subsection{\Hs-algebras}

There is a remarkable generalization of the Wedderburn structure
theorem to an infinite-dimensional setting, in a classic paper from
1945 by Warren Ambrose on `\Hs-algebras'~\cite{ambrose:hstar}. He
defines an \Hs-algebra\footnote{The notion termed 
2-\Hs-algebra in~\cite{baez:twohilbertspaces} was inspired by
Ambrose's notion of \Hs-algebra. The former could be seen as a
categorification of the latter; the two notions should not be confused.} as a 
(not necessarily unital) Banach algebra based on a Hilbert space $H$,
such that for each $x \in H$ there is an $x^* \in H$ with
\[ 
  \inprod{xy}{z} = \inprod{y}{x^*z} 
\]
for all $y,z \in H$, and similarly for right multiplication.
Note that 
\[ 
    \inprod{xy}{z} 
  = (\mu \after (x \tensor y))^\dag \after z 
  = (x^\dag \tensor y^\dag) \after \mu^\dag \after z,
\]
where we identify points $x \in H$ with morphisms $x \colon
\field{C} \to H$, and similarly
\[ 
    \inprod{y}{x^*z} 
  = y^\dag \after \mu \after (x^* \tensor z).
\]
Using the monoidal well-pointedness of $\Cat{Hilb}$, it is easy to see that
this is equivalent to the (H) condition!\footnote{Notice that neither
  axiom (H) nor Ambrose's definition of H*-algebra requires the
  operation $a \mapsto a^*$ to be continuous. However, in the setting
  of $\Cat{Hilb}$, continuity follows automatically from
  axiom~(H)~\cite[Theorem~2.3]{ambrose:hstar}.}  
% \subsection{(M) implies Banach algebra structure and properness}
The following two lemmas show that the assumptions (A),
(C), (M) and (H) indeed result in an H*-algebra.

\begin{lemma}
  A monoid in $\Cat{Hilb}$ satisfying (M) is a Banach algebra.
\end{lemma}
\begin{proof}
  The condition (M) implies that $P = \mu^\dag \after \mu$ is a
  projector: 
  \[ 
      P^2
    = \mu^\dag \after \mu \after \mu^\dag \after \mu 
    = \mu^\dag \after \mu 
    = P 
  \]
  and clearly $P = P^\dag$.
  Hence a monoid in $\Cat{Hilb}$ satisfying (M) is a Banach algebra:
  \begin{align*}
        \| xy \|^2 
    & = \inprod{xy}{xy} \\
    & = (x^\dag \tensor y^\dag) \after \mu^\dag \after \mu \after (x \tensor y) \\
    & = \inprod{x \tensor y}{P(x \tensor y)} \\
    & \leq \inprod{x \tensor y}{x \tensor y} \\
    & = \inprod{x}{x}\inprod{y}{y} \\
    & = \| x \|^2 \| y \|^2 .
  \qedhere
  \end{align*}
\end{proof}

\subsubsection*{Remark}
In fact, it can be shown that the multiplication of a semigroup in
$\Cat{Hilb}$ satisfying (H) is automatically continuous, so that
after adjusting by a constant, the semigroup is a Banach
algebra~\cite[Corollary~2.2]{ingelstam:realhstar}.\footnote{Hence 
Proposition~\ref{prop:hstaralgebra} and Theorem~\ref{thm:ambrose} can
be altered to show that a monoid in $\Cat{Hilb}$ satisfying
properness, (A), (C), and (H) (but not necessarily (M)!), corresponds
to an \emph{orthogonal} basis. This may have consequences for attempts
to classify multipartite entanglement according to various Frobenius
structures~\cite{coeckekissinger:multipartite}. Compare also the
second entry in the table on page 11 of
\cite{coeckepavlovicvicary:bases}: in finite dimension, $\delta$ is
monic by (U), but in infinite dimension, one has to explicitly
postulate $\delta$ to be monic to prevent \eg the trivial
algebra $\delta(a)=0$ and obtain a correspondence with orthogonal
bases. \label{footnote:orthogonal}} 

The following lemma establishes \emph{properness}, which
corresponds to $x^*$ being the unique vector with the property
defining H*-algebras. It follows that $R(x^*)$ is the adjoint of $R(x)$.  

\begin{lemma}
  Suppose $\delta \colon A \to A \tensor A$ in $\Hilb$ satisfies
  (A) and (H). Then (M) implies properness, \ie $aA=0
  \Rightarrow a=0$. Hence (M) holds if and only if the regular
  representation is monic. 
\end{lemma}
\begin{proof}
  By~\cite[Theorem~2.2]{ambrose:hstar}, $A$ is the direct sum of its
  trivial ideal $A'$ and a proper H*-algebra $A''$. Here, the trivial
  ideal is $A'=\{a \in A \mid aA=0\}$. Since the direct sum of Hilbert
  spaces is a dagger biproduct, we can write $\delta$ as $\delta' \oplus \delta''
  \colon A \to A \tensor A$, where $\delta' \colon A' \to A' \tensor
  A'$ and $\delta'' \colon A'' \to A'' \tensor A''$. The latter two
  morphisms are again dagger monic as a consequence of (M). So the
  multiplication $\delta'^\dag$ of $A'$ is epic, which forces $A'=0$.
\end{proof}

The following proposition summarizes the preceding discussion.

\begin{proposition}
\label{prop:hstaralgebra}
  Any structure $(A, \mu)$ in $\Cat{Hilb}$ satisfying (A), (H) and (M) is
  an \Hs-algebra (and also satisfies (F)); and conversely, an
  \Hs-algebra satisfies (A), (H) and (M), and hence also (F). 
  \qed
\end{proposition}

Ambrose proved a complete structure theorem for H*-algebras, of which
we now state the commutative case.

\begin{theorem}[Ambrose, 1945] \label{thm:ambrose}
 Any proper commutative \Hs-algebra (of arbitrary dimension) is
  isomorphic to a Hilbert space direct sum of one-dimensional
  algebras. 
  \qed
\end{theorem}

This is equivalent to asserting isomorphism
qua coalgebras. So it is exactly the result we are after! 
Rather than relying on Ambrose's results, we now give a direct,
conceptual proof, using a few notions from Gelfand duality for
commutative Banach algebras.

\subsection{Copyables and semisimplicity}

A \emph{copyable element} of a semigroup $\delta \colon A \to A \tensor A$ in
a monoidal category is a semigroup homomorphism to it from the
canonical semigroup on the monoidal unit. More precisely, a copyable element is
a morphism $a \colon I \to A$ such that $(a \tensor a) \after \delta =
\delta \after a$. In a monoidally well-pointed category such as $\Cat{Hilb}$, we
can  speak of a copyable element
of $\delta$ as a point $a \in A$ with $\delta(a)=a \tensor a$.\footnote{Copyable elements
are also called \emph{primitive} in the context of
C*-bigebras~\cite{hofmann:duality}, and \emph{grouplike} in the
study of Hopf algebras~\cite{sweedler:hopfalgebra,kassel:quantumgroups}.} 

\begin{proposition}
\label{lem:independent} 
  Assuming only (A), nonzero copyable elements are linearly independent.
\end{proposition}
\begin{proof}
  \cite[Theorem~10.18(ii)]{hofmann:duality}
 Suppose that $\{a_0,\ldots,a_n\}$ is a minimal nonempty linearly
  dependent set of nonzero copyables. Then we can write $a_0$ as $\sum_{i=1}^n
  \alpha_i a_i$ for a suitable choice of coefficients $\alpha_i \in \field{C}$. So
  \begin{align*}
        \sum_{i=1}^n \alpha_i (a_i \tensor a_i) 
    & = \sum_{i=1}^n \alpha_i \delta(a_i) \\
    & = \delta(a_0) \\
    & = (\sum_{i=1}^n \alpha_i a_i) \tensor (\sum_{j=1}^n \alpha_j a_j) \\
    & = \sum_{i,j=1}^n \alpha_i \alpha_j (a_i \tensor a_j).
  \end{align*}
  By minimality, $\{a_1,\ldots,a_n\}$ is linearly independent.
  Hence $\alpha_i^2=\alpha_i$ for all $i$, and $\alpha_i \alpha_j=0$
  for $i \neq j$. So $\alpha_i=0$ or $\alpha_i=1$ for all $i$.
  If $\alpha_j=1$, then $\alpha_i=0$ for all $i \neq j$, so
  $a_0=a_j$. By minimality, then $j=1$ and $\{a_0,a_j\} = \{a_0\}$,
  which is impossible.
  So we must have $\alpha_i=0$ for all $i$. But then $a_0=0$, which is
  likewise a contradiction.
\end{proof}

\begin{proposition}
\label{lem:normal}
  Assuming only (M), nonzero copyable elements have unit norm.
\end{proposition}
\begin{proof}
  Let $a$ be a copyable element. Then $\delta(a) = a \tensor a$.
  Hence $$\|a\| = \|\delta(a)\| = \|a
  \tensor a\| = \|a\|^2.$$ It follows that $\|a\|$ is either 0 or
  1. Therefore, if $a$ is a nonzero, then $\|a\|=1$.
\end{proof}

\begin{proposition}
\label{prop:orthogonal}
  Assuming only (F), copyable elements are pairwise orthogonal.
\end{proposition}
\begin{proof}
  \cite[Corollary~4.7]{coeckepavlovicvicary:bases}
  Let $a,b$ be copyables. Then:
  \begin{align*}
        \inprod{a}{a} \cdot \inprod{a}{a} \cdot \inprod{b}{a}
    & = \inprod{a \tensor a \tensor b}{a \tensor a \tensor a} \\
    & = \inprod{(\delta \tensor \id)(a \tensor b)}
               {(\id \tensor \delta)(a \tensor a)} \\
    & = \inprod{a \tensor b}{(\delta^\dag \tensor \id) \after
                             (\id \tensor \delta)(a \tensor a)} \\
    & = \inprod{a \tensor b}{(\id \tensor \delta^\dag) \after
                             (\delta \tensor \id)(a \tensor a)} \\
    & = \inprod{(\id \tensor \delta)(a \tensor b)}
               {(\delta \tensor \id)(a \tensor a)} \\
    & = \inprod{a \tensor b \tensor b}{a \tensor a \tensor a} \\
    & = \inprod{a}{a} \cdot \inprod{b}{a} \cdot \inprod{b}{a}.
  \end{align*}
  Analogously $\inprod{b}{b}\inprod{b}{b}\inprod{a}{b} =
  \inprod{b}{b}\inprod{a}{b}\inprod{a}{b}$. Hence, if $\inprod{a}{a}$
  and $\inprod{b}{a}$ are both nonzero, then
  $\inprod{a}{a}=\inprod{b}{a}$ and $\inprod{b}{b}=\inprod{a}{b}$. So
  $\inprod{a}{a}, \inprod{b}{b} \in \field{R}$ and $\inprod{a}{a}=
  \inprod{a}{b} = \inprod{b}{a} = \inprod{b}{b}$.
  Now suppose $\inprod{a}{b}\neq 0$. Then we can conclude 
  $\inprod{a-b}{a-b} = \inprod{a}{a} - \inprod{a}{b} - \inprod{b}{a} +
  \inprod{b}{b} = 0$. So $a-b=0$, \ie $a=b$. 
  Hence the copyables are pairwise orthogonal.
\end{proof}

%\subsection{Copyables and Characters}
Applying dagger, a copyable element of $A$ corresponds
exactly to a comonoid homomorphism $(\field{C}, \delta_{\field{C}})
\to (A, \delta)$: 
\[\xymatrix{
  1 \ar@{|->}[d] \ar@{|->}[r] & 1 \tensor 1 \ar@{|->}[d] \\  
  a \ar@{|->}[r] & a \tensor a.
}\]
We have already seen that copyable elements correspond exactly to
algebra homomorphisms  
\[ 
  (A, \mu) \to (\field{C}, \mu_{\field{C}}),
\]
\ie to \emph{characters} of the algebra---the elements of the
Gelfand spectrum of $A$~\cite{pedersen:analysisnow}. 
%\subsection{Semisimplicity}
This leads to our
first characterization of when a (nonunital) Frobenius algebra in
$\Cat{Hilb}$ corresponds to an orthonormal basis.

\begin{theorem}
  A Frobenius algebra in $\Hilb$ admits the structure theorem and
  hence corresponds to an orthonormal basis if and only if it is
  semisimple. 
\end{theorem}
\begin{proof}
  We first consider sufficiency. Form a direct sum of
  one-dimensional coalgebras indexed by the copyables of $(A,\delta)$.
  This will have an isometric embedding as a coalgebra into $(A,\delta)$: 
  \[
    e \colon \bigoplus_{\{a \mid \delta(a)=a \tensor a\}} (\field{C},
    \delta_{\field{C}}) \to (A,\delta).
  \]
  The image $S$ of $e$ is a closed subspace of $A$, and has an orthonormal
  basis given by the images of the characters of $A$ qua copyables.
  The structure theorem holds if the image of $e$ spans $A$.
  
  Given $a \in A$ and a character $c$, the evaluation $c(a)$ gives the
  Fourier coefficient of $a$ at the basis element of $S$ corresponding
  to $c$. Now $S$ will be the whole of $A$ if and only if distinct
  vectors have distinct projections on $S$, \ie if and only if
  distinct vectors have distinct Gelfand transforms $\hat{a} \colon c
  \mapsto c(a)$. Hence the Ambrose structure theorem holds when the
  Gelfand representation is injective, which holds if and only if the
  algebra is semisimple. 

  Necessity is easy to see from the form of a direct sum of
  one-dimensional algebras, as the lattice of ideals is a complete
  atomic boolean algebra, where the atoms are the generators of the algebras.
\end{proof}

% \subsection{Relating semisimplicity to (H)}
We shall restate the previous theorem in terms of axiom (H), so that we
have a characterization that lends itself to categories other than
$\Cat{Hilb}$. 

\begin{proposition}
\label{prop:semisimple}
  A Frobenius algebra in $\Hilb$ satisfies (H) if and only if it is
  semisimple, and hence admits the structure theorem. 
\end{proposition}
\begin{proof}
  Semisimplicity of proper \Hs-algebras follows from results
  in~\cite{ambrose:hstar}.  
  %If (H) is satisfied, then the algebra embeds as a subalgebra of
  %$\mathcal{B}(\HH)$ closed under adjoints, hence is a $C^*$-algebra,
  %hence is semisimple. 
  Conversely,  $\bigoplus_I (\field{C}, \mu_{\field{C}})$ is easily seen
  to satisfy (H); we can define $x^*$ by taking conjugate coefficients
  in the given basis. 
\end{proof}

\subsection{Categorical formulation}

%Apart from the algebraic formulation as a structure theorem we have
%considered so far, w
We can  recast these results into a categorical
form. Recall that there is a functor $\ell^2 \colon \Cat{PInj} \to
\Cat{Hilb}$ on the category of sets and partial
injections~\cite{barr:algebraicallycompact, heunen:thesis}. It sends a
set $X$ to the Hilbert space $\ell^2(X) = \{ \varphi \colon X \to
\field{C} \mid \sum_{x \in X} |\varphi(x)|^2 < \infty \}$, which is  
the free Hilbert space on $X$ that is equipped with an orthonormal
basis, \ie an \Hs-algebra, in a sense we will now make
precise. 
First, we make Frobenius algebras and \Hs-algebras into 
categories. While other choices of morphisms can fruitfully be
made~\cite{heunen:copyables}, the following one suits our current
purposes.   

\begin{definition}
\label{def:morphisms}
  Let $\cat{D}$ be a symmetric monoidal dagger category.
  We denote by $\Cat{HStar}(\cat{D})$ the category whose objects are
  \Hs-algebras in $\cat{D}$, and by $\Cat{Frob}(\cat{D})$ the category
  whose objects are Frobenius algebras in $\cat{D}$. 
  A morphism $(A,\delta) \to (A',\delta')$ in both categories is a
  morphism $f \colon A \to A'$ in $\cat{D}$ 
  satisfying $(f \tensor f) \after \delta = \delta' \after f$ and
  $f^\dag \after f = \id$. 
\end{definition}

\begin{proposition}
  Every object in $\Cat{PInj}$ carries a unique H*-algebra structure,
  namely $\delta(a) = (a,a)$. 
\end{proposition}
\begin{proof}
  Let $\delta = (\xymatrix@1{A & \;D\; 
  \ar@{ >->}|-{\delta_1}[l] \ar@{ >->}|-{\delta_2}[r] & A \times A})$
  be an object of $\Cat{HStar}(\Cat{PInj})$. Because of (M), we may
  assume that $\delta_1 = \id$. By (C), we find that $\delta_2$ is a
  tuple of some $d \colon A \to A$ with itself. It follows from (A)
  that  
  %   \begin{align*}
  %         (d \after d(a), d \after d(a), d(a)) 
  %     & = (\tuple{d}{d} \times \id) \after \tuple{d}{d}(a) \\
  %     & = (\id \times \tuple{d}{d}) \after \tuple{d}{d}(a) \\
  %     & = (d(x), d \after d(a), d \after d(a)),
  %   \end{align*}
  $d = d \after d$. Finally, since $\Cat{PInj}$ is monoidally well-pointed,
  $\delta$ satisfies (F) by Lemma~\ref{lem:HAWPimplyF}. Writing out
  what (F) means gives
  %   \begin{align*}
  %     &   \{((d(u),u),(u,d(u))) \mid u \in X \} \\
  %     & = \{((d(u),d(u),v),(u,v) \mid u,v \in X \} \after
  %         \{((x,y),(x,d(y),d(y)) \mid x,y \in X \} \\
  %     & = (\delta^\dag \times \id) \after (\id \times \delta) \\
  %     & = (\id \times \delta^\dag) \after (\delta \times \id) \\
  %     & = \{((u,d(v),d(v)),(u,v)) \mid u,v \in X \} \after
  %         \{((x,y),(d(x),d(x),y)) \mid x,y \in X \} \\
  %     & = \{((v,d(v)),(d(v),v)) \mid v \in X \}
  %   \end{align*}
  \[
       \{((d(b),b),(b,d(b))) \mid b \in A \} 
     = \{((c,d(c)),(d(c),c)) \mid c \in A \}.
  \]
  Hence for all $b \in A$, there is $c \in A$ with $b=d(c)$ and
  $d(b)=c$. Taking $b=d(a)$ we find that $c=a$, so that for all $a \in
  A$ we have $d \after d(a)=a$. Therefore $d = d \after d = \id$.
  We conclude that $\delta$ is the diagonal function $a \mapsto (a,a)$.
\end{proof}

As a corollary one finds that an object in $\Cat{PInj}$ with its
unique H*-algebra structure is unital if and only if it is a singleton
set, which is another good argument against demanding (U). 

If we drop the condition $f^\dag \after f = \id$ on morphisms in
Definition~\ref{def:morphisms}, the previous proposition can also be
read as saying that the categories $\Cat{HStar}(\Cat{PInj})$,
$\Cat{Frob}(\Cat{PInj})$, and $\Cat{PInj}$ are isomorphic. 
 % To establish the isomorphism of categories, we exhibit an inverse $F$
  % to the forgetful functor $U \colon \Cat{Frob}(\Cat{PInj}) \to 
  % \Cat{PInj}$ given by $(A,\delta) \mapsto A$ on objects and $f
  % \mapsto f$ on morphisms. Define $F(A) = (A,\Delta)$ on objects, and
  % $F(f)=f$ on morphisms. One easily verifies that $F$ is a
  % well-defined functor, \ie that every morphism of $\Cat{PInj}$
  % commutes with diagonals.  Unicity of \Hs-algebra structure
  % guarantees that $F \after U = \Id$ and $U \after F = \Id$. 

Since the Hilbert space $\ell^2(X)$ comes with a chosen basis induced
by $X$, the $\ell^2$ construction is in fact a functor $\ell^2 \colon
\Cat{HStar}(\Cat{PInj}) \to \Cat{HStar}(\Cat{Hilb})$. Conversely,
there is a functor $U$ in the other direction taking an \Hs-algebra to
the set of its copyables; this is functorial
by~\cite[Example~3]{ambrose:hstar}. These two functors are adjoints:  
\[\xymatrix@C+5ex{
    \Cat{HStar}(\Cat{PInj}) 
    \ar@{}|-{\perp}[r] \ar@<1ex>^-{\ell^2}[r]
  & \Cat{HStar}(\Cat{Hilb}). \ar@<1ex>^-{U}[l]
}\]
The Ambrose structure theorem, Theorem~\ref{thm:ambrose}, can now be
restated as saying that this adjunction is in fact an equivalence. 

Similarly, there is an adjunction between $\Cat{Frob}(\Cat{PInj})$ and
$\Cat{Frob}(\Cat{Hilb})$, but it is not yet clear if this is an
equivalence, too, \ie if $\Cat{Frob}(\Cat{Hilb})$ and
$\Cat{HStar}(\Cat{Hilb})$ are equivalent categories. In fact, this
question is the central issue of the rest of this paper, and will lead
to the main open question in Section~\ref{subsec:mainquestion} to
follow. In the meantime, we shall use the categorical formulation to give
different characterizations of when Frobenius algebras in
$\Cat{Hilb}$ admit the structure theorem.

\subsection{Further conditions}
\label{subsec:furtherconditions}

There are in fact a number of conditions on Frobenius algebras in
$\Hilb$ which are equivalent to admitting the structure theorem. This
section gives two more. 

\begin{theorem}
  A Frobenius algebra in $\Cat{Hilb}$ is an H*-algebra, and hence
  corresponds to an orthonormal basis, if and only if it is a
  directed colimit (in $\Cat{Frob}(\Cat{Hilb})$) of unital Frobenius
  algebras. 
\end{theorem}
\begin{proof}
  Given an orthonormal basis $\{\ket{i}\}_{i \in I}$ for $A$, define
  $\delta \colon A \to A \tensor A$ by (continuous linear extension of)
  $\delta \ket{i} = \ket{ii}$. For finite subsets $F$ of $I$, 
  define $\delta_F \colon \ell^2(F) \to \ell^2(F) \tensor \ell^2(F)$
  by $\delta_F\ket{i} = \ket{ii}$. These are well-defined objects
  of $\Cat{Frob}(\Cat{Hilb})$ by Theorem~\ref{cpvth}. Since $F$ is
  finite, every $\delta_F$ is a unital Frobenius algebra in
  $\Cat{Hilb}$. Together they form a (directed) diagram in
  $\Cat{Frob}(\Cat{Hilb})$ by inclusions $i_{F \subseteq F'} \colon
  \ell^2(F) \hookrightarrow \ell^2(F')$ if $F \subseteq F'$; the
  latter are well-defined morphisms since 
  $\delta_{F'} \after i_{F \subseteq F'} \ket{i} = \ket{ii} =
  (i_{F \subseteq F'} \tensor i_{F \subseteq F'}) \after \delta_F \ket{i}$.
  Finally, we verify that $\delta$ is the colimit of this diagram.
  The colimiting cocone is given by the inclusions $i_F \colon
  \ell^2(F) \hookrightarrow A$; these are morphisms $i_F \colon
  \delta_F \to \delta$ in $\Cat{Frob}(\Cat{Hilb})$ since $\delta \after i_F =
  (i_F \tensor i_F) \after \delta_F$, that are easily seen to form a
  cocone. Now, if $f_F \colon \delta_F \to (A',\delta')$ form another
  cocone, define $m \colon X \to X'$ by (continuous linear extension
  of) $m \ket{i} = f_{\ell^2(\{\ket{i}\})} \ket{i}$. Then $m \after
  i_F \ket{i} = f_{\ell^2(\{\ket{i}\})} \ket{i} = f_F \ket{i}$ for $i
  \in F$, so that indeed $m \after i_F = f_F$. Moreover, $m$ is the
  unique such morphism. Thus $\delta$ is indeed a colimit of the
  $\delta_F$. 

  Conversely, suppose $(A,\delta)$ is a colimit of some diagram
  $d \colon \cat{I} \to \Cat{Frob}(\Cat{Hilb})$. We will show that
  the nonzero copyables form an orthonormal basis for $A$. By
  Lemma~\ref{lem:independent} and Proposition~\ref{prop:orthogonal},
  it suffices to prove  
  that they span a dense subspace of $A$. Let $a \in A$ be given. Since
  the colimiting cocone morphisms $c_i \colon A_i \to A$ are jointly
  epic, the union of their images is dense in $A$, and therefore $a$
  can be written as a limit of $c_i(a_i)$ with $a_i \in A_i$ for some
  of the $i \in \cat{I}$.
  These $a_i$, in turn, can be written as linear combinations of
  elements of copyables of $A_i$ by Theorem~\ref{cpvth}. Now, $c_i$
  maps copyables into copyables, and so we have written $a$ as a limit
  of linear combinations of copyables of $A$. Hence the copyables of
  $A$ spans a dense subspace of $A$, and therefore form an orthonormal
  basis.  

  Finally, we verify that these two constructions are mutually
  inverse. Starting with a $\delta$, one obtains $E=\{e \mid \delta(e)
  = e \tensor e\}$, and then $\delta' \colon A \to A \tensor A$
  by (continuous linear extension of) $\delta'(e)=e \tensor e$ for
  $e \in E$. The definition of $E$ then gives $\delta' = \delta$.
  
  Conversely, starting with an orthonormal basis $\{\ket{i}\}_{i \in I}$,
  one obtains a map \mbox{$\delta \colon A \to A \tensor A$} by (continuous linear
  extension of) $\delta \ket{i} = \ket{ii}$, and then it follows that $E = \{a \in
  A \mid \delta(a) = a \tensor a\}$. It is trivial that $\{\ket{i} \mid i
  \in I\} \subseteq E$. Moreover, we know that $E$ is linearly
  independent by \ref{lem:independent}. Since it contains a
  basis, it must therefore be a basis itself. Hence indeed $E = \{ \ket{i}
  \mid i \in I\}$.
\end{proof}

For \emph{separable} Hilbert spaces, there is also a characterization
in terms of approximate units as follows.

\begin{theorem}
\label{thm:approximateunit}
  A Frobenius algebra on a separable Hilbert space in $\Cat{Hilb}$ is
  an \Hs-algebra, and hence corresponds to an orthonormal basis, if
  and only if there is a sequence $e_n$ such that $e_n a$ converges to
  $a$ for all $a$, and $(\id \tensor a^\dag) \after \delta(e_n)$ converges.
\end{theorem}
\begin{proof}
  Writing $a^*_{n} = (\id \tensor a^\dag) \after \delta(e_n)$, by
  assumption $a^* = \lim_{n \to \infty} a^*_n$ is well-defined. Since
  morphisms in $\Cat{Hilb}$ are continuous functions and composition
  preserves continuity, (H) holds by the following argument.
  \[
    
 \begin{tikzpicture}[]%
  \matrix[matrix of math nodes]
{                  & |(A)|                             \\
		 & |(B)[dot]|                        \\
|(C)| {\phantom{a^*}} &            & |(D)[box]| {a^*} \\
};
\draw (A) to (B);
\draw (B) to [out=180, in=90] (C.north) -- (C.south);
\draw (B) to [out=0, in=90] (D);
 \end{tikzpicture}%

    \; = \;
    \lim_{n \to \infty}\left( 
 \begin{tikzpicture}[]%
  \matrix[matrix of math nodes]
{          &            &       & |(B)|              \\
   |(A)[box]|{a^\dag} &&       & |(C)[dot]|         \\
   |(D)| &            & |(E)| &            & |(F)| \\
	 & |(G)[dot]|                              \\
	 & |(H)[box]|{e_{n}} &       &       & |(I)|      \\
};
\draw (G) to [out=180, in=270] (D) -- (A);
\draw (G) to [out=0, in=270] (E) to [out=90, in=180] (C);
\draw (C) to [out=0, in=90] (F) -- (I);
\draw (C) to (B);
\draw (G) to (H);

 \end{tikzpicture}%
 \right)
    \; \stackrel{\text{(F)}}{=} \;
    \lim_{n \to \infty}\left( 
 \begin{tikzpicture}[]%
  \matrix[matrix of math nodes]
{    |(A)[box]|{a^\dag} &            & |(B)|{\phantom{a^\dag}} \\
		     & |(C)[dot]|                          \\ 
		     & |(D)[dot]|                          \\         
   |(E)[box]|{e_{n}}   &            & |(F)|                  \\
};
\draw (A) to [out=270, in=180] (C);
\draw (C) to [out=0, in=270] (B.south) -- (B.north);
\draw (C) to (D);
\draw (E) to [out=90, in=180] (D);
\draw (D) to [out=0, in=90] (F.north) -- (F.south);
 \end{tikzpicture}%
 \right)
    \; = \;
    
 \begin{tikzpicture}[]%
  \matrix[matrix of math nodes]
{ |(A)| {\phantom{a^\dag}} &            & |(B)[box]| {a^\dag} \\
		   & |(C)[dot]|                           \\
		   & |(D)|                                \\
};
\draw (C) to [out=0, in=270] (B);
\draw (C) to [out=180, in=270] (A.south) -- (A.north);
\draw (C) to (D);
 \end{tikzpicture}%

  \]
  Hence approximate units imply (H). Conversely, using
  the Ambrose structure theorem, Theorem~\ref{thm:ambrose}, we can
  always define $e_n$ to be the sum of the first $n$ copyables.
  % assuming them to be well-ordered without loss of generality. 
\end{proof}

Summarizing, we have the following result.

\begin{theorem}
\label{thm:whenFimpliesH}
  For a Frobenius algebra in $\Cat{Hilb}$, the following are equivalent:
  \begin{enumerate}
    \item[(a)] it is induced by an orthonormal basis;
    \item[(b)] it admits the structure theorem;
    \item[(c)] it is semisimple;
    \item[(d)] it satisfies axiom (H);
    \item[(e)] it is a directed colimit (with respect to isometric
      homomorphisms) of finite-dimensional unital Frobenius algebras.
  \end{enumerate}
  Moreover, if the Hilbert space is separable, these are equivalent to:
  \begin{enumerate}
    \item[(f)] it has a suitable form of approximate identity.
  \qed
  \end{enumerate}
\end{theorem}

We see that the finite-dimensional result follows immediately from our
general result and Lemma~\ref{lem:cpv}, which shows that the algebra
is C* and hence semisimple. In fact, the influential
thesis~\cite{abrams:thesis} (see also~\cite{abrams:euler,kock:frobenius}) already
observes explicitly (and in much wider generality) that:  
\begin{itemize}
  \item If (M) holds, a unital Frobenius algebra is
    semisimple~\cite[Theorem~2.3.3]{abrams:thesis}.  
  \item A commutative semisimple unital Frobenius algebra  is a direct
    sum of fields~\cite[Theorem~2.2.5]{abrams:thesis}. 
\end{itemize}
Thus the only additional ingredient required to obtain
Theorem~\ref{cpvth} is the elementary Proposition~\ref{prop:orthogonal}. 

\subsection{The main question}
\label{subsec:mainquestion}

The main remaining question in our quest for a suitable notion of
algebra to characterize orthonormal bases in arbitrary dimension is the following.
\begin{center}
  \fbox{In the presence of (A), (C), and (M), does (F) imply (H)?}
\end{center}
We can ask this question for the central case of $\Hilb$, and for monoidal dagger categories in general.

If the answer is positive, then nonunital Frobenius algebras give us the
right notion of observable to use in categorical quantum mechanics. 
If it is negative, we may consider adopting (H) as the right
axiom instead of (F).

At present, these questions remain open, both for $\Hilb$ and for the
general case. However, we have been able to achieve positive results for a large
family of categories;  these will be described in the following section. 
We shall conclude this section by further narrowing down the question in the
category $\Cat{Hilb}$. 
 
Recall that the \emph{Jacobson radical} of a commutative ring is the
intersection of all its maximal regular ideals, and that a ring is
called \emph{radical} when it equals its Jacobson radical. 

\begin{proposition}
  Frobenius algebras $A$ in $\Cat{Hilb}$ decompose as a direct sum 
  \[
    A \cong S \oplus R
  \]
  of (co)algebras, where $S$ is an \Hs-algebra and $R$ is a radical
  algebra. 
\end{proposition}
\begin{proof}
  Let $a$ be a copyable element of a Frobenius algebra $A$ in
  $\Cat{Hilb}$. Consider the embedding into $A$ of $a$ as a
  one-dimensional algebra. This embedding is a kernel, since it is
  isometric and its domain is finite-dimensional. Observe that this
  embedding is an algebra homomorphism as well as a coalgebra
  homomorphism, because copyables are idempotents by (M). Now it
  follows  from~\cite[Lemma~19]{heunen:copyables} that also the
  orthogonal complement of the embedding is both an algebra
  homomorphism and a coalgebra homomorphism. Finally, Frobenius
  algebra structure restricts along such embeddings
  by~\cite[Proposition~9]{heunen:copyables}.

  We can apply this to the embedding of the closed span of all
  copyables of $A$, and conclude that $A$ decomposes (as a
  (co)algebra) into a direct sum of its copyables and the orthogonal
  subspace. By definition, the former summand is semisimple, and is
  hence a \Hs-algebra by Proposition~\ref{prop:semisimple}. The
  latter summand by construction has no copyables and hence no
  characters, and is therefore radical.  
\end{proof}

This shows how the Jacobson radical of a Frobenius algebra sits
inside it in a very simple way. Indeed, we are left with not just a
nonsemisimple algebra, but a radical one, which is the opposite of a
semisimple algebra---an algebra is semisimple precisely when its
Jacobson radical is zero. Therefore, in the category $\Cat{Hilb}$,
our main remaining question above reduces to finding out whether $R$
must be zero, as follows.
\begin{center}
  \fbox{Does there exist a nontrivial radical Frobenius algebra?}
\end{center}
Although there is an extensive literature about commutative radical
Banach algebras, including a complete classification that in fact ties
in with approximate units~\cite{esterle:radical}, this question seems
to be rather difficult.

\section{\Hs-algebras in categories of relations and positive
  matrices} 
\label{sec:hstarinrel}

We have been able to give a complete analysis of nonunital Frobenius
algebras in several
(related) cases, including:  
\begin{itemize}
  \item categories of relations, and locally bifinite relations,
    valued in cancellative quantales;
  \item nonnegative matrices with $\ell^2$-summable rows and columns. 
\end{itemize}
The common feature of these cases can be characterized as  \emph{the absence of destructive interference}.

The main result we obtain is as follows.

\begin{theorem}
  Nonunital Frobenius algebras in all these categories decompose as
  direct sums of abelian groups, and satisfy (H). 
\end{theorem}

The remainder of this section is devoted to the proof of this
theorem. Our plan is as follows. First, we shall prove the result
for $\Rel$, the category of sets and relations. 
In this case, our main question is already answered directly by
Proposition~\ref{prop:cptimpliesunital} and
Theorem~\ref{thm:whenFimpliesH}. Moreover, the result in
this case has appeared in \cite{pavlovic:frobeniusinrel}. However, our
proof is quite different, and  in particular makes no use of
units. This means that it can be carried over to the other situations
mentioned above. 

\subsection{Frobenius algebras in $\Rel$ and $\LBFRel$}
\label{frobrelsec}

We assume given a set $A$, and a Frobenius algebra structure on it
given by a relation $\Delta \subseteq A \times (A \times A)$. We shall
write $\nabla$ for $\Delta^\dag$. 

\begin{definition}
  Define $x \sim y$ if and only if $(x,y) \nabla z$ for some z. 
  By (M), the relation $\nabla$ is single-valued and
  surjective. Therefore, we may also use multiplicative notation $xy$
  (suppressing the $\nabla$), and write  $x \sim y$ to mean that $xy$ is
  defined.  
\end{definition}

\begin{lemma}
\label{lem:reflexivity}
  The relation $\sim$ is reflexive.
\end{lemma}
\begin{proof}
  Let $a \in A$. By (M), we have $a=a_1a_2$ for some
  $a_1,a_2 \in A$. 
  Then $(a_2,a) (\id \tensor \Delta) (a_2,a_1,a_2)$ and $(a_2,a_1,a_2)
  (\nabla \tensor \id) (a,a_2)$ by (C), so by
  (F) we have $(a_2,a) \Delta\after \nabla (a,a_2)$,
  so that $aa_2$ is defined.  
  % (And $a \tensor a_2$ is a fixed point of $\Delta \after \nabla$.)
  Diagrammatically, we annotate the lines with elements to show they
  are related by that morphism.
  
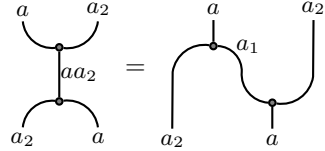
\[
    
 \begin{tikzpicture}[font=\small]%
  \matrix[matrix of math nodes]
{    |(A)[label=above:a]|  &            & |(B)[label=above:a_2]| \\
			 & |(C)[dot]|                     \\ \\
			 & |(D)[dot]|                         \\
   |(E)[label=below:a_2]| &            & |(F)[label=below:a]|  \\
};
\begin{scope}[auto]
  \draw (A) to [out=90, in=180, bend right] (C);
  \draw (C) to [out=0, in=90, bend right] (B);
  \draw (C) to node{$aa_2$} (D);
  \draw (E) to [bend left] (D);
  \draw (D) to [bend left] (F);
\end{scope}%
 \end{tikzpicture}%

    \; = \;
    
 \begin{tikzpicture}[font=\small]%
  \matrix[matrix of math nodes]
{          & |(A)[label=above:a]| &&& |(B)[label=above:a_2]| \\
	 & |(C)[dot]| &       &                    \\
   |(D)| &            & |(E)| &            & |(F)| \\
	 &            &       & |(G)[dot]|         \\
   |(H)[label=below:a_2]| &&& |(I)[label=below:a]| \\
};
\begin{scope}[auto]
  \draw (C) to [bend right] (D) -- (H);
  \draw (C) to (A);
  \draw (C) to [bend left] node{$a_1$} (E) to [bend right] (G);
  \draw (G) to (I);
  \draw (G) to [bend right] (F) -- (B);
\end{scope}
 \end{tikzpicture}%

  \]
  Also $(a,a) (\Delta \tensor \id) (a_1,a_2,a)$ and $(a_1,a_2,a) (\id
  \tensor (\Delta \after \nabla)) (a_1,a,a_2)$, so by
  (F) we have $(a,a) (\id \tensor \Delta) \after
  \Delta \after \nabla (a_1,a,a_2)$, so that $a^2$ is defined. 
  
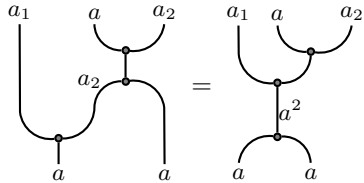
\[
    
 \begin{tikzpicture}[font=\small]%
  \matrix[matrix of math nodes]
{ |(A)[label=above:a_1]| & & |(B)[label=above:a]| & &
  |(C)[label=above:a_2]| \\
  & & & |(D)[dot]| \\
  & & & |(E)[dot]| \\
  |(F)| & & |(G)| & & |(H)| \\
  & |(I)[dot]| \\
  & |(J)[label=below:a]| & & & |(K)[label=below:a]| \\
};
\begin{scope}[auto]
  \draw (I) to [bend left] (F) -- (A);
  \draw (I) to (J);
  \draw (E) to [bend left] (H) -- (K);
  \draw (B) to [bend right] (D);
  \draw (C) to [bend left] (D);
  \draw (D) to (E);
  \draw (E) to [bend right] node[swap]{$a_2$} (G) to [bend left] (I);
\end{scope}
 \end{tikzpicture}%

    \; = \;
    
 \begin{tikzpicture}[font=\small]%
  \matrix[matrix of math nodes]
{  |(A)[label=above:a_1]| & |(B)[label=above:a]| & &
   |(C)[label=above:a_2]| \\
   |(D)| & & |(E)[dot]| \\
   & |(F)[dot]| \\ \\
   & |(G)[dot]| \\
   |(H)[label=below:a]| & & |(I)[label=below:a]| \\
};
\begin{scope}[auto]
  \draw (H) to [bend left] (G);
  \draw (I) to [bend right] (G);
  \draw (G) to node[swap]{$a^2$} (F);
  \draw (F) to [bend left] (D) -- (A);
  \draw (F) to [bend right] (E);
  \draw (E) to [bend left] (B);
  \draw (E) to [bend right](C);
\end{scope}
 \end{tikzpicture}%

  \]
  That is, $a \sim a$.
\end{proof}

\begin{lemma}
\label{lem:transitivity}
  The relation $\sim$ is transitive.
\end{lemma}
\begin{proof}
  Suppose that $a \sim b$ and $b \sim c$. Then $d=ab$ is defined.
  By Lemma~\ref{lem:reflexivity}, then $(ab)d=d^2$ is defined.
  Hence by (A), also $a(bd)$ is defined.
  Applying (F) now yields $\bar{b}$ such that $a=d\bar{b}$.
  \[  
    
 \begin{tikzpicture}[font=\small]%
  \matrix[matrix of math nodes]
{          & |(A)[label=above:d]| &&& |(B)[label=above:d]| \\
    & |(C)[dot]| &       &                    \\
     |(D)| &            & |(E)| &            & |(F)| \\
    &            &       & |(G)[dot]|         \\
     |(H)[label=below:a]| &&& |(I)[label=below:bd]| \\
};
\begin{scope}[auto]
  \draw (C) to [bend right] (D) -- (H);
  \draw (C) to (A);
  \draw (C) to [bend left] node{$b$} (E) to [bend right] (G);
  \draw (G) to (I);
  \draw (G) to [bend right] (F) -- (B);
\end{scope}
 \end{tikzpicture}%

    \; = \;
    
 \begin{tikzpicture}[font=\small]%
  \matrix[matrix of math nodes]
{    |(A)[label=above:d]| &&& |(B)[label=above:d]| \\
    &            &       & |(C)[dot]|         \\
     |(D)| &            & |(E)| &            & |(F)| \\
    & |(G)[dot]|                              \\
    & |(H)[label=below:a]| &&& |(I)[label=below:bd]| \\
};
\begin{scope}[auto]
  \draw (G) to [bend left] (D) -- (A);
  \draw (A) -- (D);
  \draw (G) to [bend right] node[swap]{$\bar{b}$} (E) to [bend left] (C);
  \draw (C) to [bend left] (F) -- (I);
  \draw (C) to (B);
  \draw (G) to (H);
\end{scope}
 \end{tikzpicture}%

  \]
  It now follows from (C) that
  $a=d\bar{b}=a\bar{b}b$; in particular $a\bar{b}$ is defined. But
  then also $ac=(a\bar{b}b)c=(a\bar{b})(bc)$ is seen to be defined by the
  assumption $b \sim c$ and another application of (F). 
  \[
    
 \begin{tikzpicture}[font=\small]%
  \matrix[matrix of math nodes]
{    |(A)[label=above:a\bar{b}]| && |(B)[label=above:bc]| \\
    & |(C)[dot]|         \\ \\
    & |(D)[dot]|         \\         
     |(E)[label=below:a]| && |(F)[label=below:c]| \\
};
\begin{scope}[auto]
  \draw (A) to [bend right] (C);
  \draw (C) to [bend right] (B);
  \draw (C) to node{$ac$} (D);
  \draw (E) to [bend left] (D);
  \draw (D) to [bend left] (F);
\end{scope}%
 \end{tikzpicture}%

    \; = \;
    
 \begin{tikzpicture}[font=\small]%
  \matrix[matrix of math nodes]
{    |(A)[label=above:a\bar{b}]| &&& |(B)[label=above:bc]| \\
    &            &       & |(C)[dot]|         \\
     |(D)| &            & |(E)| &            & |(F)| \\
    & |(G)[dot]|                              \\
    & |(H)[label=below:a]| &&& |(I)[label=below:c]| \\
};
\begin{scope}[auto]
  \draw (G) to [bend left] (D) -- (A);
  \draw (A) -- (D);
  \draw (G) to [bend right] node[swap]{$b$} (E) to [bend left] (C);
  \draw (C) to [bend left] (F) -- (I);
  \draw (C) to (B);
  \draw (G) to (H);
\end{scope}
 \end{tikzpicture}%

  \]
  Hence $a \sim c$.
\end{proof}

\begin{proposition}
\label{prop:disjointunion}
  The Frobenius algebra $A$ is a disjoint union of totally defined
  commutative semigroups, each satisfying (F). 
\end{proposition}
\begin{proof}
  By the previous two lemmas and (C), the
  relation $\sim$ is a equivalence relation. Hence $A$ is a disjoint
  union of the equivalence classes under $\sim$. By definition of
  $\sim$, the multiplication $\nabla$ is totally defined on these
  equivalence classes. Moreover, they inherit the 
  properties (M), (C),
  (A) and (F) from $A$.
\end{proof}

%By Proposition~\ref{prop:himpliesf}, the same holds for \Hs-algebras: every
%\Hs-algebra in $\Cat{Rel}$ is a disjoint union of \Hs-semigroups. So
%to prove that (F) implies (H),  
%it suffices to consider the case of (commutative) semigroups. 
% A commutative semigroup satisfies (F) if and only if
% \[
%   ab=cd \Longleftrightarrow \exists_{e,f}[ae=c \wedge ed=b \wedge cf=a
%   \wedge fb=d].
% \]
% And a commutative semigroup satisfies (H) if and only if
% \[
%   \forall_a \exists_{a^*} \forall_{b,c}[ ab=c \Longleftrightarrow ca^*=b ].
% \]

\begin{lemma}
\label{lem:huntington}
  \cite{huntington:group}
  A semigroup $S$ is a group if and only if $aS=S=Sa$ for all $a \in S$.
\end{lemma}
\begin{proof}
  The condition $aS=S$ means $\forall b\exists c[b=ac]$.
  If $S$ is a group, this is obviously fulfilled by $c=a^{-1}b$. 
  For the converse, fix $a \in S$. Applying the condition with $b=a$
  yields $c$ such that $a=ac$. Define $e=c$, and let $x \in S$. Then
  applying the condition with $b=x$ gives $c$ with $x=ac$. Hence
  $ex=eac=ac=x$. Thus $S$ is a monoid with (global) unit $e$.
  Applying the condition once more, with $a=x$ and $b=e$ yields
  $x^{-1}$ with $xx^{-1}=e$.
\end{proof}

\begin{theorem}
\label{lem:group}
  $A$ is a disjoint union of commutative groups. 
\end{theorem}
\begin{proof}
Let $A'$ be one of the equivalence classes of $A$, and
  consider $a,b \in A'$. This means  that $a \sim b$. As in the proof of
  Lemma~\ref{lem:transitivity}, there is a $\bar{b}$ such that
  $a=\bar{b}ba$. Putting $c=\bar{b}a$ thus gives $\forall a, b \in A'
  \exists c \in A'[a=cb]$. In other words, $aA'=A'$ and similarly $A'=A'a$ for
  all $a \in A$. Hence $A'$ is a (commutative) group by
  Lemma~\ref{lem:huntington}. 
\end{proof}

The following theorem already follows from
Proposition~\ref{prop:cptimpliesunital} and
Theorem~\ref{thm:whenFimpliesH}, but now we have a direct proof that
also carries over to the theorem after it, which does not follow from
the earlier results.

\begin{theorem}
\label{thm:characterisationFinRel}
  In $\Rel$, Frobenius algebras satisfy (H), and the conditions (F)
  and (H) are equivalent in the presence of the other axioms for
  Frobenius algebras. 
\end{theorem}
\begin{proof}
  This  follows directly from Lemma~\ref{lem:group}, since we can
  define $a^*=a^{-1}$, where $a^{-1}$ is the inverse in the disjoint
  summand containing $a$. 
  More precisely, a point of $A$ in $\Rel$ will be a subset of $A$,
  and we apply the definition $a^*=a^{-1}$ pointwise to this
  subset. This assignment is easily seen to satisfy (H). 
\end{proof}

\begin{theorem}
  In $\LBFRel$, Frobenius algebras are disjoint unions of abelian
  groups and hence satisfy (H), and the conditions (F) and (H) are
  equivalent in the presence of the other axioms for Frobenius
  algebras.  
\end{theorem}
\begin{proof}
  The proof above made no use of units, and is equally valid in
  $\LBFRel$. 
\end{proof}

\subsection{Quantale-valued relations}

We shall now consider categories of the form $\Rel(Q)$, where $Q$ is a
commutative, cancellative quantale. Recall that a commutative quantale
\cite{rosenthal:quantales} is a structure $(Q, {\cdot}, 1, {\leq})$,
where $(Q, {\cdot}, 1)$ is a commutative monoid, and $(Q, {\leq})$ is
a partial order which is a complete lattice, \ie it has suprema of
arbitrary subsets. In particular, the supremum of the empty set is the
least element of the poset, written $0$. The multiplication is
required to distribute over arbitrary joins, \ie 
\[ 
    x \cdot (\bigvee_{i \in I} y_i) 
  = \bigvee_{i \in I} x \cdot y_i, 
  \qquad 
    (\bigvee_{i \in I} x_i) \cdot y 
  = \bigvee_{i \in I} x_i \cdot y. 
\]
The quantale is called \emph{cancellative} if
\[ 
  x \cdot y = x \cdot z
  \; \Rightarrow \; 
  x = 0 \vee y = z. 
\]
An example is given by the extended nonnegative reals $[0, \infty]$
with the usual ordering, and multiplication as the monoid operation. 
Note that the only nontrivial example when the monoid operation is
idempotent, \ie when the quantale is a locale, is the two-element
boolean algebra $\cat{2} = \{ 0, 1 \}$, since in the idempotent case $x \cdot 1
= x \cdot x$ for all $x$. We write $\CQ$ for the category of
cancellative quantales which are nontrivial, \ie in which $0 \neq 1$.  

\begin{proposition}
\label{qtermprop}
  The two element boolean algebra is terminal in $\CQ$.
\end{proposition}
\begin{proof}
  The unique homomorphism $h \colon Q \to \cat{2}$ sends $0$ to
  itself, and everything else to 1. 
  Preservation of sups holds trivially, and cancellativity implies
  that multiplication is preserved. 
\end{proof}

The category $\Rel(Q)$ has sets as objects; morphisms $R \colon X
\rrel Y$ are $Q$-valued matrices, \ie functions $X \times Y \to
Q$. Composition is relational composition evaluated in $Q$, \ie if $R 
\colon X \rrel Y$ and $S \colon Y \rrel Z$, then
\[ 
  S \after R(x, z) = \bigvee_{y \in Y} R(x, y) \cdot S(y, z). 
\]
It is easily verified that this yields a category, with identities
given by diagonal matrices; that it has a monoidal structure induced
by cartesian product; and that is has a dagger given by matrix
transpose, \ie relational converse. 
Thus $\Rel(Q)$ is a symmetric monoidal dagger category, and the notion
of Frobenius algebra makes sense in it. Note that $\Rel(\cat{2})$ is
just $\Rel$.% \footnote{But relations in sheaf categories are in
% general not of the form $\Rel(Q)$.}

A homomorphism of quantales $h \colon Q \to R$ induces a (strong)
monoidal dagger functor $h^* \colon \Rel(Q) \to \Rel(R)$, which
transports Frobenius algebras in $\Rel(Q)$ to Frobenius algebras in
$\Rel(R)$. 
In particular, by Proposition~\ref{qtermprop}, a Frobenius algebra
$\Delta \colon A \rrel A \times A$ in $\Rel(Q)$ has a reduct $h^*
\Delta \colon A \rrel A \times A$ in $\Rel$. Hence
Theorems~\ref{lem:group} and \ref{thm:characterisationFinRel} apply to 
this reduct. 
The remaining degree of freedom in the Frobenius algebra in $\Rel(Q)$
is which elements of $Q$ can be assigned to the elements of the
matrix. 

Suppose that we have a Frobenius algebra $\Delta \colon A \rrel A
\times A$ in $\Rel(Q)$. We write $M \colon (A \times A) \times A \to Q$
for the matrix function corresponding to the converse of $\Delta$, and
we write $M(a, b, c)$ rather than $M((a, b), c)$. 

Because the unique homomorphism $Q \to \cat{2}$ reflects 0, an entry
$M(a, b, c)$ is nonzero if and only if the corresponding relation
$(a, b) (h^* \nabla) c$ holds. Applying Theorem~\ref{lem:group}, this
immediately implies that for each $a, b \in A$, there is exactly
one $c \in A$ such that $M(a, b, c) \neq 0$. 

We can use this observation to apply similar diagrams to those used in
our proofs for $\Rel$ to obtain constraints on the values
taken by the matrix in $Q$. 

\begin{proposition}
\label{quantvalprop}
  With notation as above:
  \begin{enumerate}
    \item[(a)] If $e$ is an identity element in one of the disjoint
      summands, then for all $a,b$ in that disjoint summand we have
      $M(a, e, a) = M(b, e, b)$. We write $q_e$ for this common value.  
    \item[(b)] For all $a,b \in A$, we have $M(a, b, ab)^2 = 1$.
  \end{enumerate}
\end{proposition}
\begin{proof}
  For (a), consider the diagram
  
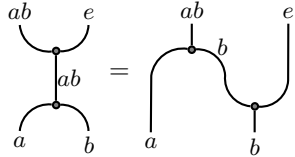
\[
    
 \begin{tikzpicture}[font=\small]%
  \matrix[matrix of math nodes]
{    |(A)[label=above:ab]|  &            & |(B)[label=above:e]| \\
	     & |(C)[dot]|                     \\ \\
	     & |(D)[dot]|                         \\
     |(E)[label=below:a]| &            & |(F)[label=below:b]|  \\
};
\begin{scope}[auto]
  \draw (A) to [out=90, in=180, bend right] (C);
  \draw (C) to [out=0, in=90, bend right] (B);
  \draw (C) to node{$ab$} (D);
  \draw (E) to [bend left] (D);
  \draw (D) to [bend left] (F);
\end{scope}
 \end{tikzpicture}%

    \; = \;
    
 \begin{tikzpicture}[font=\small]%
  \matrix[matrix of math nodes]
{          & |(A)[label=above:ab]| &&& |(B)[label=above:e]| \\
    & |(C)[dot]| &       &                    \\
     |(D)| &            & |(E)| &            & |(F)| \\
    &            &       & |(G)[dot]|         \\
     |(H)[label=below:a]| &&& |(I)[label=below:b]| \\
};
\begin{scope}[auto]
  \draw (C) to [bend right] (D) -- (H);
  \draw (C) to (A);
  \draw (C) to [bend left] node{$b$} (E) to [bend right] (G);
  \draw (G) to (I);
  \draw (G) to [bend right] (F) -- (B);
\end{scope}
 \end{tikzpicture}%

  \]
  This implies the equation
  \[ 
      M(a, b, ab) \cdot M(ab, e, ab) 
    = M(a, b, ab) \cdot M(b, e, b) 
  \]
  and hence, by cancellativity, $M(ab, e, ab) = M(b, e, b)$. Hence for
  any $c$, taking $b = a^{-1}c$, $M(c, e, c) = M(b, e, b)$. 

  For (b), consider the diagram
  
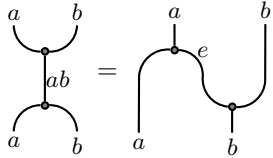
\[
    
 \begin{tikzpicture}[font=\small]%
  \matrix[matrix of math nodes]
{    |(A)[label=above:a]|  &            & |(B)[label=above:b]| \\
	     & |(C)[dot]|                     \\ \\
	     & |(D)[dot]|                         \\
     |(E)[label=below:a]| &            & |(F)[label=below:b]|  \\
};
\begin{scope}[auto]
  \draw (A) to [out=90, in=180, bend right] (C);
  \draw (C) to [out=0, in=90, bend right] (B);
  \draw (C) to node{$ab$} (D);
  \draw (E) to [bend left] (D);
  \draw (D) to [bend left] (F);
\end{scope}
 \end{tikzpicture}%

    \; = \;
    
 \begin{tikzpicture}[font=\small]%
  \matrix[matrix of math nodes]
{          & |(A)[label=above:a]| &&& |(B)[label=above:b]| \\
    & |(C)[dot]| &       &                    \\
     |(D)| &            & |(E)| &            & |(F)| \\
    &            &       & |(G)[dot]|         \\
     |(H)[label=below:a]| &&& |(I)[label=below:b]| \\
};
\begin{scope}[auto]
  \draw (C) to [bend right] (D) -- (H);
  \draw (C) to (A);
  \draw (C) to [bend left] node{$e$} (E) to [bend right] (G);
  \draw (G) to (I);
  \draw (G) to [bend right] (F) -- (B);
\end{scope}
 \end{tikzpicture}%

  \]
  This implies the equation $M(a, b, ab)^2 = q_e^2$.
  Now applying (M), for each $c$ we obtain that
  \[ 
    \bigvee \{ M(a, b, ab)^2  \mid ab = c \} = 1. 
  \]
  As all the terms in this supremum are the same, $M(a, b, ab)^2 = 1$.  
\end{proof}

Thus if $q^2 = 1$ implies $q = 1$ in $Q$, the matrix $M$ is in fact
valued in $\cat{2}$. Otherwise, we can choose square roots of unity for
the entries. 

\begin{theorem} 
  Let $Q$ be a cancellative quantale. Suppose that $q^2 = 1$ implies
  $q = 1$ in $Q$. Then every Frobenius algebra in $\Rel(Q)$
  satisfies (H). 
  \qed
\end{theorem}

\subsection{Positive $\ell^2$ matrices}

We now consider the case of matrices in $\LMat$ valued in
the non-negative reals. These form a monoidal dagger subcategory of
$\LMat$, which we denote by $\LMpos$. 
Note that the semiring $(\Rpos, +, 0, {\times}, 1)$ has a
unique 0-reflecting semiring homomorphism to $\cat{2}$. Hence a
Frobenius algebra in $\LMpos$ has a reduct to one in $\Rel$ via this
homomorphism. Just as before, we can apply Theorem~\ref{lem:group} to
this reduct. 

We have the following analogue to Proposition~\ref{quantvalprop},
where $M \colon (A \times A) \times A \to \Rpos$ is the matrix realizing
the Frobenius algebra structure. 

\begin{proposition}
\label{Mconstprop}
  The function $M$ is constant on each disjoint summand of $A$.
\end{proposition}
\begin{proof}
  We can use the same reasoning as in
  Proposition~\ref{quantvalprop}(a) to show that, if $e$ is an
  identity element in one of the disjoint summands, then for all $a$,
  $b$ in that disjoint summand, $M(a, e, a) = M(b, e, b)$. We write
  $r_e$ for this common value. 

  Using the same reasoning as in  Proposition~\ref{quantvalprop}(b)
  one finds $M(a, b, ab)^2 = r_e^2$. Since we are in $\Rpos$, this implies
  $M(a, b, ab) = r_e$, so that $M$ is constant on each disjoint
  summand.  
\end{proof}

\begin{proposition}
  Each disjoint summand is finite, and the common value of $M$ on that
  summand is $1/\sqrt{d}$, where $d$ is the cardinality of the
  summand. 
\end{proposition}
\begin{proof}
  Applying (M), for each $c$ in the summand we obtain that
  \[ 
    \sum_{ab = c} M(a, b, c)^2  = 1. 
  \]
  Since the summand is a group, for each $c$ and $a$ there is a unique
  $b$ such that $ab = c$. 
  Moreover, by Proposition~\ref{Mconstprop}, all the terms in this sum
  are equal. Thus the sum must be finite, with the number of terms $d$
  the cardinality of the summand. We can therefore rewrite the
  equation as $d r_e^2 = 1$, and hence $r_e = 1/\sqrt{d}$.
\end{proof}

\begin{theorem}
  Every Frobenius algebra in $\LMpos$ satisfies (H).
\end{theorem}
\begin{proof}
  We make the same pointwise assignment  $x^* = x^{-1}$ on the
  elements of $A$ as in the proof of
  Theorem~\ref{thm:characterisationFinRel}, with the weight
  $1/\sqrt{d}$ determined by the summand. 
\end{proof}

In the case when the matrix represents a bounded linear map in
$\Hilb$, we can apply Theorem~\ref{thm:whenFimpliesH}, and obtain the
following.

\begin{proposition}
  If a Frobenius algebra in $\Hilb$ can be represented by a non-negative
  real matrix, then it corresponds to a direct sum of one-dimensional
  algebras, and hence to an orthonormal basis. 
  \qed
\end{proposition}

Conversely, if a Frobenius algebra in $\Cat{Hilb}$ satisfies (H), it
is induced by an orthonormal basis, and hence has a matrix
representation with nonnegative entries. Therefore we have found
another equivalent characterization of when (F) implies (H) in
$\Cat{Hilb}$ to add to our list in Theorem~\ref{thm:whenFimpliesH}.

\begin{proposition}
  A Frobenius algebra $A$ in $\Cat{Hilb}$ satisfies (H) if
  and only if there is a basis of $A$ such that the matrix of the
  comultiplication has nonnegative entries when represented on that basis.
\end{proposition}

\subsection{Discussion}

How different is the situation with Frobenius algebras in these
matrix categories from $\Hilb$? In fact, it is not as different as it
might at first appear.
\begin{itemize}
  \item The category $\Hilb$ is equivalent to a full subcategory of
    the dagger monoidal category of complex matrices with
    $\ell^2$-summable rows and columns. The `only' assumption needed
    but not satisfied is positivity. 
    %It might yet be possible to tweak the result for matrix
    %categories to get one for $\Hilb$. 
  \item The result `looks' different, but beware. Consider \emph{group
    rings} (or algebras) over the complex numbers, for finite
    abelian groups. They can easily be set up to fulfil all our
    axioms, including (U), so that Theorem~\ref{cpvth} applies, and they
    decompose as direct sums of one-dimensional algebras.
    But the isomorphism which gives this decomposition may be quite
    non-obvious.\footnote{It would be interesting, for example, to know the computational complexity of determining this isomorphism, given a presentation of the group. As far as we know, this question has not been studied.} Note that copyable elements are idempotents, so the
    only group element which is copyable is the identity. 
  \item This decomposition result indeed shows that group rings over
    the complex numbers are very weak invariants of the groups. The
    group rings of two finite abelian groups will be isomorphic if the
    groups have the same order~\cite{strichartz:groupalgebras}! 
  \item However, this is highly sensitive to which field we are
    over. Group algebras over the \emph{rationals} are isomorphism
    invariants of groups~\cite{ingelstam:realhstar}. 
\end{itemize}

\section{Outlook}
\label{sec:outlook}

We are still investigating our main question, of whether (F) implies
(H), in $\Hilb$ and elsewhere. 

Beyond this, we see the following main lines for continuing a
development of categorical quantum mechanics applicable to
infinite-dimensional situations. 

\begin{itemize}
  \item We are now able to consider observables with infinite discrete
    spectra. Beyond this lie continuous observables and
    projection-valued measures; it remains to be seen how these can be
    analyzed in the setting of categorical quantum mechanics. 
  \item Complementary observables should be studied in this
    setting. The bialgebra approach studied in
    \cite{coeckeduncan:observables} is based on axiomatizing
    \emph{mutually unbiased bases}, and does not extend directly to
    the infinite-dimensional case. However, complementary observables
    are studied from a much more general perspective in works such as
    \cite{buschetal:operational}, and this should provide a good basis
    for suitable categorical axiomatizations. 
  \item This leads on to another point.  There may be other means,
    within the setting of categorical quantum mechanics, of
    representing observables, measurements and complementarity, which
    may be more flexible than the Frobenius algebra approach, and in a
    sense more natural, since tensor product structure is not inherent
    in the basic notion of measurement. Methodologically, one should
    beware  of concluding over-hastily that a particular approach is
    canonical, simply on the grounds that it captures the standard
    notion in finite-dimensional Hilbert spaces. There may be several
    ways of doing this, and some more definitive characterization
    would be desirable. 
  \item A related investigation to the present one is the work on
    nuclear and traced ideals in \cite{abramskyetal:nuclear}. 
    It seems likely that some combination of the ideas developed
    there, and those we have studied in this paper, will prove
    fruitful. 
\end{itemize}

\subsubsection*{Acknowledgement}

We thank Rick Blute for stimulating the early stages of the
research that led to this article.

\bibliographystyle{amsalpha}
\bibliography{frobcliff}

\end{document}